\documentclass[conference,letterpaper]{IEEEtran}

\usepackage[utf8]{inputenc} 
\usepackage[T1]{fontenc}
\usepackage{url}
\usepackage{ifthen}
\usepackage{cite}
\usepackage[cmex10]{amsmath} % Use the [cmex10] option to ensure complicance
\usepackage{tikz}
\usetikzlibrary{quantikz2,decorations}
\tikzset{every node/.style={font=\scriptsize}} %controls the font sizes of ALL tikzpictures
\usepackage{comment}
\usepackage{cleveref}
\usepackage{amsmath}
\usepackage{float}
\usepackage{amssymb}
\usepackage{amsfonts}
\usepackage{thmtools}
\declaretheorem{definition}
\declaretheorem[sibling=definition]{theorem}

\declaretheorem[sibling=definition]{proposition}
\declaretheorem[sibling=definition]{lemma}

\usepackage{enumitem}
\crefname{enumi}{part}{parts}
\newcommand{\dd}{\mathrm{d}}
\DeclareMathOperator{\tr}{Tr}

\DeclareMathOperator{\Vol}{Vol}
\DeclareMathOperator{\interior}{int}
\DeclareMathOperator{\Euc}{Euc}

\allowdisplaybreaks
\def \Pbb {\mathbb{P}}
\def \Ebb {\mathbb{E}}
\def \Xcal {\mathcal{X}}

\definecolor{darkgreen}{rgb}{0,0.5,0}
\usepackage{graphicx}
\allowdisplaybreaks
\begin{document}
\title{Optimal Error Exponents for Composite Sequential Quantum Hypothesis Testing\\
\thanks{STJ has received funding from EPSRC Quantum Technologies Career Acceleration Fellowship (UKRI1218).}}

% %%% Single author, or several authors with same affiliation:
% \author{%
%  \IEEEauthorblockN{Author 1 and Author 2}
% \IEEEauthorblockA{Department of Statistics and Data Science\\
%                    University 1\\
 %                   City 1\\
  %                  Email: author1@university1.edu}% }

%%% Several authors with up to three affiliations:
\author{%
  \IEEEauthorblockN{{Jacob Paul Simpson, Efstratios Palias and Sharu Theresa Jose}}
  \IEEEauthorblockA{School of Computer Science,
                    University of Birmingham, UK\\
                    Email: jps538@student.bham.ac.uk,
                    e.palias@bham.ac.uk,
                    s.t.jose@bham.ac.uk}}

\maketitle

%%%%%%
%% Abstract: 
%% If your paper is eligible for the student paper award, please add
%% the comment "THIS PAPER IS ELIGIBLE FOR THE STUDENT PAPER
%% AWARD." as a first line in the abstract. 
%% For the final version of the accepted paper, please do not forget
%% to remove this comment!
%%

\begin{abstract}
    We study the composite sequential quantum hypothesis testing (SQHT) problem, where the objective is to distinguish a null quantum state from a set of alternative quantum states. We propose a mixture-sequential quantum probability ratio test that adaptively selects measurements based on the current mixture estimate of the  alternative set, and stops upon the first threshold crossing of the mixture log-likelihood ratio. Under an expected sample size constraint, we show that our proposed strategy simultaneously achieves the Type-I and (worst-case) Type-II error exponents,  characterized by the minimal measured relative entropies between the null state and the alternative set. We further establish a matching converse, thereby characterizing the  optimal error exponent region. Finally, our results show that achieving vanishing error probabilities in composite SQHT requires an expected sample complexity at least as large as that of sequential testing between two fixed states. 
\end{abstract}

\section{Introduction}\label{sec:introduction}
Quantum hypothesis testing (QHT) is a fundamental statistical inference problem in quantum information theory with applications in quantum sensing, cryptography, and communications \cite{barndorff2003quantum,bae2015quantum,pirandola2018advances,bergou2010discrimination}. Given $n$ copies of an unknown quantum state, the \emph{simple} binary QHT problem is to design a two-outcome measurement that 
%In the standard setting of simple binary QHT, a decision-maker is given $n$ copies of an unknown quantum state and must design a two-outcome measurement to
determines whether the state is $\rho^{\otimes n}$ (null state) or $\sigma^{\otimes n}$ (alternative state) \cite{Helstrom1969,Holevo1973}. However, in many practically relevant tasks, such as quantum state verification \cite{Thinh2020worstcase}, the alternative consists of a set of quantum states. This leads to \textbf{composite QHT}, where the goal is to distinguish a null state $\rho^{\otimes n}$ from an alternative set $\mathcal C_n$ of quantum states. A canonical example  is $\mathcal{C}_n=\{\sigma^{\otimes n}: \sigma \in \mathcal{D}\}$, the set of i.i.d. tensor-product states generated from a compact, convex set $\mathcal{D}$, although more general correlated alternatives are also of interest \cite{berta2021composite, hayashi2025generalized}.
%with the goal of distinguishing a null state $\rho^{\otimes n}$ from a composite alternative set $\mathcal{C}_n$ of quantum states.
Performance of the test is evaluated by two types of errors: the Type-I error, corresponding to incorrectly rejecting the null hypothesis when it is true, and the Type-II error, corresponding to  incorrectly accepting the null when the true state belongs to the alternative set.
%measured in the worst-case sense over all states in $\mathcal{C}_n$. 
A central goal of QHT is to characterize the optimal trade-off between these errors.\enlargethispage{-.03cm}

In the fixed-length (i.e., fixed-$n$) setting, this trade-off is well-understood asymptotically for composite QHT with i.i.d.  tensor-product alternatives. Specifically, when the Type-I error is constrained below a constant, the optimal exponential decay rate of the Type-II error is given by the generalized quantum Stein's exponent, $\inf_{\sigma \in \mathcal{D}}D(\rho \Vert \sigma)$, where $D(\rho \Vert \sigma)$ denotes the quantum relative entropy  \cite{berta2021composite,brandao2010generalization, berta2023gap,  hayashi2002optimal, hayashi2025generalized}. 
% \textcolor{purple}{Conversely, in the reversed formulation with composite null and simple alternate and we constrain the Type-II error, the generalized quantum Sanov exponent $\inf_{\rho \in \mathcal{D}}D(\rho \Vert \sigma)$ characterizes the exponential rate of decay of the Type-I error \cite{bjelakovic2005quantum, berta2021composite}.}
Conversely, when the Type-II error is constrained, $\inf_{\sigma \in \mathcal{D}}D(\sigma \Vert \rho)$ characterizes the exponential rate of decay of the Type-I error \cite{bjelakovic2005quantum, berta2021composite}. When both errors decrease exponentially, the optimal trade-off between their exponents is characterized by the generalized quantum Hoeffding bound \cite{fang2026error}. These results reveal a fundamental limitation of fixed-length testing: the optimal exponents for the two errors cannot, in general, be simultaneously achieved.
%both exponents cannot simultaneously attain their extremal values by the quantum Hoeffding bound 
% Moreover, the case when both exponents are required to be the same is characterised by the quantum Chernoff exponent \cite{audenaert2008asymptotic, nussbaum2009chernoff}, which is strictly smaller than the two extremal relative entropy rates in general.
%These results highlight an inherent limitation of the fixed-length formulation: 
%optimising one exponent necessarily worsens the other.

This limitation motivates \emph{sequential} quantum hypothesis testing (SQHT) \cite{slussarenko2017quantum,martinez2021quantum}, where copies are acquired on demand and measurements are performed sequentially until a data-dependent stopping criterion is met. In contrast to the fixed-length setting, recent work \cite{li2022optimal} showed that, for simple hypotheses, SQHT can simultaneously attain both error exponents under an expected sample size constraint. Specifically, these exponents are given by the measured relative entropies under adaptive single-copy measurements, and by the quantum relative entropies when collective measurements on blocks are permitted.
However, the corresponding question for composite SQHT remains open.
%In this work, we study the achievable error exponents for composite SQHT.

Extending SQHT from simple to composite hypotheses introduces fundamentally new challenges. Even in the classical setting, sequential composite testing is highly non-trivial: generalized sequential probability ratio tests (SPRTs) simultaneously achieve the optimal worst-case error exponents only under strong regularity assumptions \cite{li2014generalized,pan2022asymptotics}, while mixture-based SPRTs typically provide prior-averaged guarantees, without general worst-case guarantees \cite{wald1947sequential}. 
%In the quantum setting, these challenges are further compounded by adaptive measurement strategies, which induce non i.i.d observations across each round.

Motivated by these challenges, we study \textbf{composite SQHT under single-copy adaptive measurements} and characterize the optimal error exponent region. We propose a novel \emph{mixture-sequential quantum probability ratio test} that, under an expected sample size constraint, simultaneously achieves the Type-I and Type-II error exponents,  characterized by the minimal measured relative entropies between the null state and alternative set. We also provide a matching converse, thereby establishing the optimal error exponent region. Our results extend classical mixture-based approaches by providing worst-case, rather than prior-averaged, guarantees. Furthermore, we show that achieving  vanishing error probabilities in the composite SQHT requires at least as many expected samples as any simple SQHT instance contained in $\mathcal{D}$.
\section{Problem Setting}\label{sec:problem_setting}
\subsection{Notation}\label{sec:notation}\enlargethispage{-.05cm}
Throughout this work, we use $\rho$ and $\sigma$ to denote quantum states on the $d$-dimensional Hilbert space $\mathbb{C}^d$. These states are represented as density matrices: Hermitian, positive semi-definite, unit-trace matrices in $\mathbb{C}^{d \times d}$. Equivalently, given the affine space $\mathsf A_d := \{ \rho \in \mathbb C^{d\times d} : \rho = \rho^\dagger,\ \tr[\rho]=1 \}$, quantum states are elements of the convex subset of the positive semi-definite matrices $\mathsf S_d := \{ \rho \in \mathsf A_d : \rho \succeq 0 \}$, and we say that $\rho \in \mathsf S_d$ is full-rank if it is positive definite, i.e., $\rho \succ 0$.

Let $\mathcal X$ denote a finite set of measurement outcomes. A positive operator-valued measure (POVM) is a \emph{set} of Hermitian operators $M=\{M(x)\}_{x \in \mathcal X}$ in $\mathbb{C}^{d \times d}$ such that $M(x)\succeq0$ for all $x \in \mathcal{X}$ and $\sum_{x\in \mathcal{X}}M(x)=\mathbb{I}$. We denote the set of all POVMs with outcome set $\mathcal X$ as $\mathcal M_{\mathcal X}$, which is compact. Applying a POVM $M \in \mathcal M_{\mathcal X}$ to a quantum state $\rho\in \mathsf S_d$ induces a probability distribution $P_{\rho,M}$ on $\mathcal X$ via Born’s rule:  $P_{\rho,M}(x) = \tr[M(x)\rho]$ for $x \in \mathcal X$. If, in addition, $M(x)^2=M(x)$ for all $x\in\mathcal X$, we call $M$ a projector-valued measure (PVM).
% We further assume that the POVMs $M \in \mathcal{M}_{\mathcal{X}}$ have full support over $\mathcal{X}$, i.e. $M(x) \neq 0$ for all $x \in \mathcal{X}$. Hence, for every full-rank quantum state $\rho$, the Born distribution $P_{\rho,M}$ has full support on $\mathcal{X}$.
% We let $D(\rho\|\sigma):=\tr[\rho(\log\rho-\log\sigma)]$ be the quantum relative entropy, which reduces to the Kullback-Leibler divergence if its inputs are probability distributions on $\mathcal X$ \textcolor{purple}{and we use the convention $0\log(\tfrac{0}{0})=0$.}
Finally, throughout this paper, we use upper case $M$ to denote a \emph{random} POVM and lower case $m$ to denote its realization, while $X$ and $x$ denote a random measurement outcome and its realization, respectively. We use $m_1^k=(m_1,\ldots,m_k)$  to denote a sequence of POVMs,  and $x_1^k=(x_1,\ldots, x_k)$ to denote a sequence of measurement outcomes, and similarly $M_1^k$ and $X_1^k$ for the corresponding random sequences.

%Lastly, in this work we use uppercase letters $M_i$ and $X_i$ to denote random variables, while the lowercase letters $m_i$ and $x_i$ represent their corresponding realisations.

\subsection{Adaptive Composite Sequential  QHT}\label{sec:singleton_composite}
This work focuses on {the} composite QHT problem, where the goal is to distinguish an unknown quantum state $\rho^\star$ between the null hypothesis $H_0:\rho^\star =\rho$, where $\rho \in \mathsf S_d$ is \emph{full-rank}, and the alternative hypothesis $H_1: \rho^\star \in \mathcal D$. Here, $\mathcal D\subset \mathsf S_d$ is a compact, convex set of \emph{full-rank} quantum states with non-empty interior relative to $\mathsf A_d$, and $\rho \not \in \mathcal D$.

We adopt a \emph{sequential} and \emph{adaptive} testing framework, where measurements are selected based on past observations. To this end, we define a sequential quantum hypothesis test (SQHT) by the tuple $\mathcal{S}:=(\mathcal{X}, \{\mu_k,d_k\}_{k=1}^{\infty})$, where $\mu_k(\dd m_k|x_1^{k-1},m_1^{k-1})$ is the conditional probability measure that determines the next measurement $m_k$ based on the history of past observations, and $d_k(x_1^k,m_1^k) \in \{0,1,*\}$ is the decision rule. At time $k\geq 1$, the tester receives a fresh copy of the unknown state and selects a POVM $M_k=m_k$ according to the conditional probability measure $\mu_k(\cdot|x_1^{k-1},m_1^{k-1})$. 
%Note here that the POVM selection at time $k$ depends on the history of past observations, making the testing strategy \emph{adaptive}. 
The tester then measures the unknown state according to the chosen POVM and observes the outcome $X_k=x_k$ with probability $\tr[m_k(x_k)\rho^\star]$. Based on the observations so far $(x_1^k,m_1^k)$, the tester makes the decision $d_k(x_1^k,m_1^k) \in \{0,1,*\}$. If $d_k=0$ (or $d_k=1$), the tester stops the test and declares that the null (or the alternative) hypothesis is true. Instead, if $d_k=*$, the tester decides that there is not enough information to make the decision and continues the test by requesting another sample.

%In simple adaptive sequential binary QHT, a quantum system is prepared in one of two quantum states, either $\rho_0$ or $\rho_1$, which in this work we assume have known density matrix representations. Then at each time $k\geq 1$, the distinguisher receives a fresh copy of the unknown state and selects a measurement according to a conditional probability measure $\mu(\cdot|\mathcal{H}_{k-1})$ that depends on the \emph{history} $\mathcal{H}_{k-1}$, which is defined as the sequence of all chosen measurements and their corresponding outcomes up to time $k$. Specifically $\mathcal{H}_k = (M_1, X_1, \dots, M_k, X_k)$, where each measurement $M_j$ is chosen based on the information available in the previous history $\mathcal{H}_{j-1}$, and $X_j$ is the corresponding measurement outcome. The distinguisher then applies the realised POVM $M_k=m_k$ onto the $k$-th sample of the unknown state, which gives an outcome $X_k=x_k$. In particular, by Born's rule the conditional distribution of each outcome satisfies is \[\Pbb_{n,\rho}(X_j = x \mid \mathcal H_{j-1}) = \tr[M_j(x)\rho]\]where $\Pbb_{n,\rho}$ denotes the probability law of the history $\mathcal{H}_k$ when the unknown state is $\rho$. Afterwards, a decision rule $d_{n,k}(\mathcal{H}_k)\in\{0,1,*\}$ is evaluated. The value $*$ indicates that no decision is made, and the distinguisher continues sampling, while the values $0$ and $1$ correspond to deciding in favour of hypotheses $\rho_0$ and $\rho_1$ respectively, in which case the procedure stops. 
The first time $d_k\in \{0,1\}$ occurs is called the \emph{stopping time}, defined as $T:=\inf\{k \geq 1:d_k \in \{0,1\}\}$. The stopping time $T$ represents the \emph{random} number of samples used in the SQHT. Henceforth, we define SQHT via the updated tuple $\mathcal{S}=(\mathcal{X}, \{\mu_k,d_k\}_{k=1}^{\infty},T)$ with the stopping time.\enlargethispage{-.1cm}

%For each $m\ge1$, we define the bounded stopping time as $T^{(m)}:=T\wedge m.$
Let $\Omega:= (\mathcal{M}_{\mathcal{X}} \times \mathcal{X})^{\infty}$ be the infinite product space induced by $ \mathcal{M}_{\mathcal{X}} \times \mathcal{X}$ with the usual product topology, and let $\mathcal{F}$ be the $\sigma$-algebra generated by this topology. For a given SQHT $\mathcal{S}$ and underlying quantum state $\rho^\star\in \mathsf S_d$, we can define a probability measure $\Pbb_{\mathcal{S},\rho^\star}$ as follows: for any $k \in \mathbb{N}$, measurable set $A \in \mathcal B((\mathcal{M}_{\mathcal{X}})^k)$ and sequence $x_1^k$,
\begin{align}
    &\Pbb_{\mathcal{S},\rho{^\star}}[(A \times \{x_1^k\}) \times (\mathcal{M}_{\mathcal{X}} \times \mathcal{X})^{\infty}]\nonumber \\&=\int_A \prod_{j=1}^k\mu_j(\dd m_j|x_1^{j-1},m_1^{j-1}) \tr[m_j(x_j)\rho{^\star} ]. \label{eq:probability_measure}
\end{align}
%\textcolor{purple}{I have not changed the above but think we should as the original abuses notation. From what I understand, when they say that $A \in (\mathcal{M}_{\mathcal{X}})^k$, what they really mean is that $A \in \mathcal B((\mathcal{M}_{\mathcal{X}})^k)$ (Borel $\sigma$-algebra), since $A$ is an event and thus a measurable subset of $(\mathcal{M}_{\mathcal{X}})^k$. Next, I think the sequence of observations based of the measurement should be wrote as a singleton set $\{x_1^k\}$, currently we have $x_1^k=(x_1,\hdots, x_k)$ so $A \times x_1^k$ is not well defined.} 

The above measure then defines the random process $\{(M_k,X_k)\}_{k=1}^{\infty}$. Let $\mathcal{F}_k \subset \mathcal{F}$ be the $\sigma$-algebra generated by  $(M_1^k,X_1^k)$. Then, 
%the event $\{T=k\}$ is determined by $(M_1^k,X_1^k)$, and thus belongs to $\mathcal{F}_k$. Therefore, 
$T$ is a stopping time with respect to the filtration $\{\mathcal{F}_k\}_{k=0}^{\infty}$ with $\mathcal{F}_0 := \{\emptyset, \Omega\}$. 
%-- Need to define filtration for $k=0$ or otherwise $\mathcal{F}_{k-1}$ is undefined for $k=1$.
%In this paper, the filtration we choose is the natural filtration $\mathbb{F} = \{\mathcal{F}_k\}_{k \in \mathbb{N}}$, where $\mathcal{F}_k=\sigma(\mathcal{H}_k)$ is the smallest $\sigma$-algebra generated by the history $\mathcal{H}_k$, which represents the accumulated information up to step $k$. Then $T_n$ is a valid stopping time with respect to $\mathbb{F}$.
\subsection{Constraints and Achievable Exponents}
In the following, we consider sequences of SQHTs $\mathcal{S}_n=(\mathcal{X}, \{\mu_k,d_{n,k}\}_{k=1}^{\infty},T_n)$, indexed by $n \in \mathbb{N}$, where $\mathcal{S}_n$ has stopping time $T_n$. Let $\Pbb_{n,\rho}$ denote the probability measure in \eqref{eq:probability_measure} corresponding to $\mathcal{S}_n$ when the underlying state is $\rho$, and let $\Ebb_{n,\rho}$ denote the associated expectation.

%. Accordingly, we denote the probability measure in \eqref{eq:probability_measure} as $\Pbb_{n,\rho}$.
We are interested in studying SQHTs
under an \emph{expectation constraint} on the stopping time $T_n$: 
\begin{align}
\max\biggl\{\mathbb E_{n,\rho}[T_n], \sup_{\sigma \in \mathcal D }\mathbb E_{n,\sigma}[T_n]\biggr\}\leq n\label{eq:expectation_constraint}.
\end{align} 
%where the expectation $\Ebb_{n,\rho}$ is evaluated with respect to the probability measure $\Pbb_{n,\rho}$. 
    Under this constraint, we analyze the trade-off between the Type-I and Type-II error probabilities. The Type-I error {$\alpha_n^\star$} is the probability that the test declares $H_1$ to be true when the true state is $\rho$, i.e., $\rho^\star=\rho$:
    \begin{equation}\label{eq:typeIerror}
        \alpha^\star_n:=\Pbb_{n,\rho}(d_{T_n}=1).
    \end{equation}
In contrast, the Type-II error {$\beta_n^\star$} is the worst-case probability that the test declares $H_0$ to be true when the true state $\rho^{\star} \in \mathcal D$:
\begin{equation}\label{eq:typeIIerror}
    \beta_n^{\star}:=\sup_{\sigma\in \mathcal D}\beta_n(\sigma),\quad\beta_n(\sigma):=\Pbb_{n,\sigma}(d_{T_n}=0).
\end{equation}      
We next formalize the notion of \emph{achievable error exponents} and define the corresponding \emph{error exponent region}.
\begin{definition}[Achievable Error Exponent Pair and the Error Exponent Region]
A pair $\left(R_0,R_1\right) \in \mathbb R^2_{+}$ is said to be an \textbf{achievable error exponent pair} under the expectation constraint if there exists a sequence of SQHTs $\{\mathcal S_n\}_{n\in \mathbb N}$ satisfying \vspace{-0.6cm}
\begin{gather}
\liminf_{n\to\infty} \frac{1}{n} \log \frac{1}{\alpha_n^\star} \ge R_0, \label{condition_1}\\
\liminf_{n\to\infty} \frac{1}{n} \log \frac{1}{\beta_n^\star} \ge R_1, \quad \text{and} \label{condition_2}\\[-1mm]
\limsup_{n\to\infty} \biggl(  \max\bigl\{\mathbb{E}_{n,\rho}[T_n], \sup_{\sigma\in \mathcal D}\mathbb{E}_{n,\sigma}[T_n] \bigr\} -  n \biggr)\le 0. \label{condition_3}
\end{gather}
We then define the  \textbf{error exponent region} $\mathcal{A}(\{\rho\}, \mathcal D) \subset \mathbb{R}^2_+$ as the closure of the set of all achievable error exponent pairs under the expectation constraint.
\end{definition}
\subsection{Related Work: Adaptive Simple Sequential QHT}\label{sec:related_work}
Our composite SQHT framework generalizes the adaptive simple sequential QHT  setting studied in \cite{li2022optimal}, corresponding to the special case when $\mathcal D=\{\sigma\}$ is a singleton set. For this setting, the authors characterize the optimal error exponent region as follows.\enlargethispage{-.1cm}
\begin{theorem}
    Let $\rho,\sigma \in \mathsf S_d$ be full-rank. Then,
    \begin{align}
    \mathcal{A}{(\{\rho\},\{\sigma\})}&= \bigl\{ (R_0, R_1) :
R_0 \le D_{\mathcal{M}}(\sigma \| \rho), \nonumber \\ 
&\quad \quad \quad \quad \quad \quad  \hspace{0.12cm} R_1 \le D_{\mathcal{M}}(\rho \| \sigma) 
\bigr\},  \hspace{0.3cm}  \\ 
\mbox{where} \quad D_{\mathcal M}(\rho\|\sigma) &:= \sup_{\mathcal{X}}\sup_{m \in \mathcal M_{\mathcal X}} D\big(P_{\rho,m}\|P_{\sigma,m}\big) \nonumber \\
& = D(P_{\rho,m^*} \Vert P_{\sigma, m^*}) \label{eq:measuredentropy_def}
\end{align}
is the measured relative entropy between states $\rho$ and $\sigma$, with the supremum achieved by a PVM $m^*$ with $|\mathcal{X}|=d$ \cite[Theorem 2]{Berta2017}.
%\textcolor{red}{and by \cite[Theorem 2]{Berta2017} the supremum is achieved for some rank-1 PVM, i.e. $|\mathcal X|=d$. } 
Here, $P_{\rho,m}$ is the probability distribution on $\Xcal$ induced by $m$, i.e., $P_{\rho,m}(x):=\tr[m(x)\rho]$, and $D(P\|Q)$ is the Kullback--Leibler divergence between two probability distributions $P$ and $Q$.
\end{theorem}
\vspace{-0.1cm}
\section{Main Results}\label{sec:main_results}
This section presents our main result, which characterizes the error exponent region for adaptive composite SQHT under the expectation constraint. 

\begin{theorem}\label{theorem:big}
The error exponent region for composite SQHT under the expectation constraint is given by
    % Let $\rho$ be a quantum state with full-rank and $\mathcal D$ be a compact, convex set of full-rank quantum states with non-empty interior relative to the affine space $\{ \sigma  : \sigma  = \sigma ^\dagger, \ \tr[\sigma ]=1 \}$. Then
    \vspace{-0.2cm} 
        \begin{align}
    \mathcal{A}{(\{\rho\},\mathcal D)}&= \bigl\{ (R_0, R_1) :
R_0 \le D_{\mathcal{M}}(\mathcal D \| \rho), \nonumber \\ 
&\quad \quad \quad \quad \quad \quad  \hspace{0.12cm} R_1 \le D_{\mathcal{M}}(\rho \| \mathcal D) 
\bigr\},  \hspace{0.3cm} 
\end{align} \vspace{-0.8cm} 

\hspace{-0.5cm} where we denote
\vspace{-0.25cm} 
   \begin{align}
       D_{\mathcal {M}}(\rho\Vert\mathcal D)\hspace{-0.1cm}:= \hspace{-0.1cm}\inf_{\sigma \in \mathcal D} D_{\mathcal M}(\rho\Vert\sigma)\text{ and }
 D_{\mathcal M}(\mathcal D\Vert\rho)\hspace{-0.1cm}:=\hspace{-0.1cm} \inf_{\sigma \in \mathcal D} D_{\mathcal M}(\sigma\Vert\rho). \nonumber 
    \end{align}
\end{theorem}
\begin{comment}
\begin{figure}[!t]
    \centering
    \begin{tikzpicture}
        \draw[thick,-stealth](-.5,0)--(5.5,0)node[above]{{$\displaystyle\liminf_{n\to\infty}\frac1n\log\frac1{\alpha_n^\ast}$}};
        \draw[thick,-stealth](0,-.5)--(0,3)node[right]{{$\displaystyle\liminf_{n\to\infty}\frac1n\log\frac1{\beta_n^\ast}$}};
        \node[below left]{$0$};
        \draw[thick,blue,dashed](4,2)--(0,2)node[black,left]{$D_{\mathcal M}(\rho\|\sigma)$};
        \draw[thick,blue,dashed](4,2)--(4,0)node[black,below]{$D_{\mathcal M}(\sigma\|\rho)$};
        \draw[thick,blue](2.2,1.5)--(0,1.5)node[black,left]{$D_{\mathcal M}(\rho\|\mathcal D)$};
        \draw[thick,blue](2.2,1.5)--(2.2,0)node[black,below]{$D_{\mathcal M}(\mathcal D\|\rho)$};
        \fill[blue,opacity=.2](0,0)rectangle(2.2,1.5);
        \node at(1.1,.75){$\mathcal A(\{\rho\},\mathcal D)$};
        \node at(2.2,1.5){$\bigstar$};
    \end{tikzpicture}
        \caption{Illustration of the achievable region (shaded area) of \Cref{theorem:big}. The rates achieved by our adaptive test are denoted with ``{\scriptsize$\bigstar$}''. The dashed rectangle bounds the achievable region $\mathcal A(\{\rho\},\{\sigma\})$ for some $\sigma\in\mathcal D$.}
    \label{fig:region}
\end{figure}
\end{comment}
%{\color{red}We illustrate this error exponent region and compare it to the one corresponding to simple QHT in \Cref{fig:region}.}
\Cref{theorem:big} also yields the following first-order asymptotic sample complexity interpretation: the expected number of samples required to achieve vanishing error probabilities $\alpha^{\star},\beta^{\star}\to 0$ scales as 
\vspace{-0.2cm}
\begin{align}
\mathbb E_{\rho}[T]\approx\frac{\log(1/\beta^{\star})}{D_{\mathcal M}(\rho \| \mathcal D)}, \hspace{0.1cm} 
\sup_{\sigma \in \mathcal D}\mathbb E_{\sigma}[T]\approx\frac{\log(1/\alpha^{\star})}{D_{\mathcal M}(\mathcal D \| \rho)}. \nonumber
\end{align} Since $D_{\mathcal M}(\rho \Vert \mathcal D) \leq D_{\mathcal M}(\rho \Vert \sigma)$ and $D_{\mathcal M}(\mathcal D \Vert \rho ) \leq D_{\mathcal M}(\sigma \Vert \rho)$ for every $\sigma \in \mathcal{D}$, composite SQHT requires at least as many expected samples as the corresponding simple SQHT \cite{li2022optimal} to achieve vanishing error probabilities.
\subsection{Achievability of \Cref{theorem:big}}\label{sec:achievability}
    We now construct a sequence of  SQHTs, $\{\mathcal{S}_n\}_{n=1}^\infty$, where $\mathcal{S}_n=(\mathcal{X},\{\mu_k,d_{n,k}\}_{k=1}^{\infty},T_n)$, that achieves the error exponent pair $(D_{\mathcal M}(\mathcal D\Vert\rho), D_{\mathcal M}(\rho \| \mathcal D))$. By \eqref{eq:measuredentropy_def},   without loss of generality we let $\mathcal X=\{1,2,\dots,d\}$, and define $m_0^\ast(\sigma)=\{m_0^\ast(\sigma)(x)\}_{x\in\mathcal X}$ and $m_1^\ast(\sigma)=\{m_1^\ast(\sigma)(x)\}_{x\in\mathcal X}$ as the rank-1 PVMs  that achieve the suprema in the definitions of $D_{\mathcal M}(\rho\|\sigma)$ and $D_{\mathcal M}(\sigma\|\rho)$, respectively. To this end, we introduce a new \emph{mixture-sequential quantum probability ratio test} (mixture-SQPRT) for the  composite SQHT problem under study. This test relies on a mixture version of the Sequential Quantum Log-Likelihood Ratio (mixture-SQLLR) defined as follows. \enlargethispage{-.02cm}
%To this end, we start by defining the mixture-sequential quantum log-likelihood ratio (mixture-SQLLR).
\begin{definition}[mixture-SQLLR]\label{def:mixture-SQLLR} Let $\Pi$ be a prior probability measure on the compact, convex set $\mathcal D$ of quantum states. For each $\sigma \in \mathcal D$ and $k\geq 1$, let 
\vspace{-0.1cm}
\begin{align}
 S_k(\sigma)=\sum_{j=1}^{k} \log 
\frac{\tr[M_j(X_j)\rho]}{\tr[M_j(X_j)\sigma]}, \quad S_0(\sigma):=0 \label{eq:pointwise_SQPRT}
\end{align} denote the \emph{pointwise-SQLLR}\footnote{Note that all likelihood-ratio sums introduced here are evaluated on the common support of the induced distributions. For full-rank states, this amounts to summing only over outcomes $x\in\mathcal X$ for which $M(x)\neq 0$.} and $L_k (\sigma):=e^{-S_k(\sigma)}$. Then, the mixture-SQLLR can be defined as
\vspace{-0.1cm}
\begin{align}
\widetilde{S}_k := -\log \widetilde{L}_k, \quad \widetilde L_{k}:= \int_{\mathcal D} L_k(\sigma)\Pi(\dd\sigma). \label{eq:mixture-SQPRT}
\end{align}
\end{definition}
%\textcolor{purple}{where $\Pi$ is defined as a probability measure on $\mathcal D$, with $\Pi(\dd\sigma)$ denoting integration with respect to $\Pi$. Further details on the construction of $\Pi$ are provided in the Appendix.} 

At each time $k \geq 1$, the observed history $(M_1^k,X_1^k)$ is used to update the prior $\Pi$ on $\mathcal D$ to the posterior
\vspace{-0.1cm}
\begin{align}
\widetilde\Pi_{k}(\dd\sigma) := \frac{L_{k}(\sigma)\Pi(\dd\sigma)}{\widetilde L_{k}},
\end{align}
and define $\widetilde\Pi_{0}(\dd\sigma):=\Pi(\dd\sigma)$. Then, the corresponding \emph{barycenter} (posterior mean state) is defined as 
\vspace{-0.1cm}
\begin{equation} 
\tilde\sigma_{k} := \int_{\mathcal D} \sigma \widetilde\Pi_{k}(\dd\sigma). \label{eq:barycenter}
\end{equation} Importantly, by the convexity of $\mathcal D$, we have $\tilde\sigma_k \in \mathcal D$ almost surely. Moreover, $\tilde\sigma_k$ is $\mathcal{F}_k$-measurable. 

Using the barycenter,  the mixture-SQLLR admits the incremental representation (see \Cref{lemma:bary} for details):
\vspace{-0.1cm}
\begin{align}\label{eq:incremental}
    \widetilde{S}_k =\sum_{j=1}^k \widetilde{Z}_j, \quad \widetilde Z_j := \log \frac{\tr\big[M_j(X_j)\rho\big]}{\tr\big[M_j(X_j)\tilde\sigma_{j-1}\big]}.
\end{align}
Finally, note that when $\mathcal D$ is singleton (i.e., $\mathcal D=\{\sigma\}$) as in the simple QHT setting, the  mixture-SQLLR coincides with the pointwise-SQLLR. 
%Furthermore,  $\tilde\sigma_{k-1} \in \mathcal D$ a.s., since it is a convex combination of elements of $\mathcal D$, as well it follows that $\tilde\sigma_{k-1}$ is $\mathcal{F}_{k-1}$-measurable since $\widetilde{\Pi}_{k-1}$ depends only on $\mathcal{F}_{k-1}$. 

%Lastly, a direct calculation shows that the sequential mixture log-likelihood ratio single-step increment at time $k$ is
%\begin{equation}
% \widetilde Z_k := \log \frac{\tr\big(M_k(X_k)\rho\big)}{\tr\big(M_k(X_k)\tilde\sigma_{k-1}\big)}.
%\end{equation}
We now define the adaptive measurement selection strategy $\mu_k(\cdot|x_{1}^{k-1},m_{1}^{k-1})$ used in our mixture-SQPRT. At time $k=1$,  the tester picks the rank-1 PVM $M_1=m_0^\ast(\tilde\sigma_0)$ and for $k>1$, $\mu_k$ selects the rank-1 PVM  $M_k$ \emph{deterministically} as:
\begin{align}
M_k =
\begin{cases}
m_{0}^\ast(\tilde\sigma_{k-1}) \text{ for }
\widetilde S_{k-1}\ge0,\\[6pt]
m_{1}^\ast(\tilde\sigma_{k-1}) \text{ for }
 \widetilde S_{k-1}<0,
\end{cases} \label{eq:mu_k}
\end{align}
where $\tilde\sigma_{k-1}$ is defined as in \eqref{eq:barycenter}. These PVMs can be chosen as measurable functions of $\sigma$, ensuring that $M_k$ is $\mathcal F_{k-1}$-measurable (see \Cref{appendix:prem}).
Finally, we define the decision function $d_{{n,}k}$ at time $k$ as:
%Moreover, for every $n$, given the thresholds $A_n,B_n \in (0,\infty)$, the decision rule at time $k$ is defined as
\vspace{-0.2cm}
\begin{equation}
d_{{n,}k} (X_1^k,M_1^k) =
\begin{cases}
0 & \text{if } \widetilde S_k \ge B_{{n}}, \\
1 & \text{if } \widetilde S_k \le -A_{{n}}, \\
\ast & \text{otherwise}.
\end{cases} \label{eq:d_k}
\end{equation}

We refer to the test $(\mathcal{X}, \{\mu_k,d_{{n,}k}\}_{k=1}^{\infty},T_n)$ with $\mu_k$ and $d_{{n,}k}$ specified as in \eqref{eq:mu_k} and \eqref{eq:d_k}, respectively, as a \textbf{mixture-SQPRT with parameters $(A_{{n}},B_{{n}})$}. 
We then have the following achievability result.
\begin{proposition} \label{prop:achievability}
For $\hspace{-0.5mm}\tau \hspace{-0.5mm}\in \hspace{-0.5mm}(\hspace{-0.3mm}0,\min\left\{ \hspace{-0.4mm}D_{\mathcal M}(\rho \| \mathcal D), \hspace{-0.4mm}D_{\mathcal M}(\mathcal D\| \rho)\right\}\hspace{-0.5mm})$, let 
  \begin{align}A_n &:= n\big( D_{\mathcal M}(\mathcal D \| \rho) - \tau \big), \hspace{0.1cm} B_n := n\big( D_{\mathcal M}(\rho \| \mathcal D) - \tau \big)  \label{eq:b_n}, \end{align} 
\enlargethispage{-.06cm}
\hspace{-0.17cm} and $\Pi$ in \Cref{def:mixture-SQLLR} be the normalized volume measure on $\mathcal D$. Then, the sequence $\{(\mathcal{X}, \{\mu_k,d_{n,k}\}_{k=1}^{\infty},T_n )\}_{n=1}^{\infty}$ of mixture-SQPRTs with the parameter sequence $\{(A_n,B_n)\}_{n=1}^{\infty}$ satisfies \eqref{condition_1}, \eqref{condition_2} and \eqref{condition_3} with $R_0=D_{\mathcal M}(\mathcal D\Vert\rho)$ and $R_1=D_{\mathcal M}(\rho \| \mathcal D)$.
\end{proposition}
\subsection{Converse of \Cref{theorem:big}}\label{sec:converse}
We next show that the error exponents achieved in \Cref{prop:achievability} are optimal, i.e., they cannot be improved by \emph{any} sequence of adaptive SQHTs.%by proving an analogue of \cite[Lemma 20]{li2022optimal}.
\begin{proposition}\label{prop:converse}
    Let $\{\mathcal{S}_n\}_{n=1}^{\infty}$ be a sequence of SQHTs with adaptive strategies such that $\alpha^{\star}_n, \beta^{\star}_n \rightarrow 0$ and the sequence $\{T_n\}_{n=1}^{\infty}$ satisfies the expectation constraint \eqref{condition_3}. Then, \begin{align}
\limsup_{n\to\infty}\frac1n\log\frac1{\alpha_n^{\star}}&\leq D_{\mathcal M}(\mathcal{D}\|\rho),\text{ and}\label{eq:a}\\[0mm] \limsup_{n\to\infty}\frac1n\log\frac1{\beta_n^{\star}}&\leq D_{\mathcal M}(\rho\Vert \mathcal{D}).\label{eq:b}
    \end{align}
\end{proposition}
\begin{IEEEproof}[Proof sketch]
Fix any $\sigma\in\mathcal D$. For the Type-I error, invoking \cite[Lemma 20]{li2022optimal}, we obtain
\vspace{-0.1cm}
\begin{equation}
\limsup_{n\to\infty}\frac1n\log\frac1{\alpha_n^\star}\leq D_{\mathcal M}(\sigma\|\rho). \nonumber
\end{equation}
Since this upper bound holds for every $\sigma\in\mathcal D$, minimizing the RHS over $\sigma\in\mathcal D$ yields \eqref{eq:a}. Similarly, 
% \begin{equation}
%     \limsup_{n\to\infty}\frac1n\log\frac1{\beta_n(\sigma)}\leq D_{\mathcal M}(\rho\|\sigma).
% \end{equation}
using that $\log\frac1{\beta_n^\star}\leq\log\frac1{\beta_n(\sigma)}$, for all $\sigma\in\mathcal D$, and applying the corresponding converse bound for simple SQHT yields \eqref{eq:b}.
\end{IEEEproof}
\section{Proof of Achievability}\label{sec:proof_achievability}
This section presents a proof sketch of \Cref{prop:achievability}, and since $\tau>0$ is arbitrary, it suffices to prove the statement with $R_0=D_{\mathcal M}(\mathcal D\|\rho)-\tau$ and $R_1=D_{\mathcal M}(\rho\|\mathcal D)-\tau$. We begin by establishing several preliminary results that underpin the analysis, with proofs deferred to the Appendix.
%We start by presenting some preliminary results. These will form the foundation of our analysis and are all proven in the Appendix.
\vspace{-0.1cm}
\subsection{Preliminary Results}\label{sec:prelim}
% We begin by establishing a basic structural property of $\mathcal D$.
% \textcolor{purple}{Since the function $\lambda_{\min}(\sigma)$ is continuous in $\sigma$ and $\mathcal D$ is compact, the minimum $\min_{\sigma\in\mathcal D}\lambda_{\min}(\sigma)=: \delta$ is attained by the Extreme Value Theorem \textcolor{blue}{\cite[Theorem 4.16]{rudin1976principles}}. Moreover, as every $\sigma \in \mathcal D$ is full-rank we have $\delta>0$. 
% Because every $\sigma\in \mathcal D$ has unit trace, then $\delta\leq 1/d$ and thus $\mathcal D$ satisfies
% \begin{equation}
% \mathcal D\subseteq \left\{\sigma: \tr[\sigma]=1, \sigma\succeq \delta I\right\}.
% \end{equation}
% for some $\delta\in (0,1/d]$
% }
% then $\lambda_{\min}(\sigma)>0$}
% Let $\delta:= \min_{\sigma\in\mathcal D}\lambda_{\min}(\sigma)$, then since $\mathcal D$ is a compact and convex set of full-rank quantum states on a $d$-dimensional Hilbert space and \textcolor{purple}{$\lambda_{\min}(\sigma)$ is continuous in in $\sigma$ $\lambda_{\min}(\sigma)>0$}, then $\delta\in(0,1/d]$ and hence $\mathcal D$ satisfies
%  \begin{equation}
% \mathcal D\subseteq \left\{\sigma: \tr[\sigma]=1, \sigma\succeq \delta I\right\}.  
% \end{equation} 

We first note that the following processes are martingales.
\begin{lemma}\label{lemma:martingales}
Let $Y_k^{(\rho)}:=-\widetilde S_k+\sum_{j=1}^k\mathbb E_{n,\rho}[\widetilde Z_j \mid \mathcal{F}_{j-1}]$ with $Y_0^{(\rho)}:=0$. Then, 
%and\textcolor{blue}{, for any $\sigma\in\mathcal D$, let} {\color{red}$Y_k^{(\sigma)}:=\widetilde S_k+\sum_{j=1}^k\mathbb E_{n,\sigma}[-\widetilde Z_j \mid \mathcal{F}_{j-1}]$}{\color{red} the notation is confusing}. 
    $\{L_k(\sigma)\}_{k \geq 0}$, $\{\widetilde L_k\}_{k\ge0}$ and $\{Y_k^{(\rho)}\}_{k \geq 0}$ are martingales under $\Pbb_{n,\rho}$ with respect to the filtration $\{\mathcal{F}_k\}_{k\geq 0}$.
  %  \item  $\{Y_k^{(\sigma)}\}_{k\ge1}$ is a martingale under $\Pbb_{n,\sigma}$ with respect to the same filtration. {\color{red}  doesn't this follow from the Martingale property of $Y_k^{(\rho)}$ since $Y_k^{(\sigma)}$ is just the negative of $Y_k^{(\rho)}$ evaluated at $\rho=\sigma$?}
%\end{itemize} 
% where $\mathcal F_0:=\{\emptyset, \mathcal X^{\mathbb N}\}$ and $\mathcal F_k:=\sigma(\mathcal{H}_k)$ for $k\geq 1$. 
\end{lemma}

The following lemma uniformly bounds the deviation of mixture-SQLLR from the pointwise-SQLLR over all states $\sigma \in \mathcal D$. This will be critical in the Type-II error analysis.
\vspace{-0.4cm}
\begin{lemma}\label{lemma:mixture_martingale_comparison_with_pointwise}
\setlength{\abovedisplayskip}{3pt}\setlength{\belowdisplayskip}{3pt}%
\setlength{\abovedisplayshortskip}{3pt}\setlength{\belowdisplayshortskip}{3pt}%
Fix an auxiliary reference state $\sigma^\circ$ in the interior of $\mathcal D$ relative to $\mathsf A_d$. Then, for every $k\geq 1$ and every $\sigma\in\mathcal D$,
\begin{align}
\widetilde S_k - S_k(\sigma) \le \widetilde C + (d^2-1)\log k ,
\end{align}
where $\widetilde C$ is a finite constant that depends on the fixed choice of $\sigma^\circ$, but is independent of $k$ and $\sigma$.
\end{lemma}

In particular, \Cref{lemma:mixture_martingale_comparison_with_pointwise} shows that the worst-case gap between the mixture and pointwise-SQLLRs at time $k$ grows only logarithmically with $k$.
%\textcolor{red}{Although $\widetilde C$ depends on the  fixed auxiliary reference state $\sigma^\circ\in\operatorname{int}(\mathcal D)$, this does not affect any future arguments. The supremum is never taken over $\sigma^\circ$ so the required uniformity is only over $k\geq 1$ and $\sigma\in\mathcal D$.}
The next lemma characterizes the conditional drift of $\widetilde{S}_k$ and establishes its exponential concentration under $\Pbb_{n,\rho}$, together with uniform control of the process under the alternative measure $\Pbb_{n,\sigma}$. These properties are crucial for verifying the expectation constraint.
\begin{lemma}\label{lemma:lemma19}
For the mixture-SQPRTs, the following holds:
\begin{enumerate}[label=(\roman*)]
\item\label{lemma:19;part1}  The conditional expectation of $\widetilde Z_k$ given $\mathcal F_{k-1}$ satisfies
\vspace{-0.23cm}
\begin{multline}
\mathbb{E}_{n,\rho}[\widetilde Z_k \mid \mathcal F_{k-1}] =\mathbf 1_{\{\widetilde S_{k-1}\ge 0\}} D_{\mathcal M}(\rho\| \tilde\sigma_{k-1}) \\+\mathbf 1_{\{\widetilde S_{k-1}<0\}} D(P_{\rho,m_1^\ast(\tilde\sigma_{k-1})}\|P_{\tilde\sigma_{k-1},m_1^\ast(\tilde\sigma_{k-1})}).
\end{multline}
        In particular, 
        % there exists a constant $\textcolor{blue}{c_1}>0$ independent of $k$ and $n$ such that
    % \[
    % \mathbb{E}_{n,\rho}[\widetilde Z_k \mid \mathcal F_{k-1}] \ge \textcolor{blue}{c_1}
    % \quad \text{a.s.},
    % \]
    % and hence 
    $\{\widetilde S_k\}_{k\ge0}$ has uniformly positive conditional drift under $H_0$. 
\enlargethispage{-.04cm}
\item\label{lemma:19;part2}  For sufficiently small $\lambda>0$, there exists $c \in (0,1)$ s.t. 
\vspace{-0.2cm}
\begin{align}
\mathbb{E}_{n,\rho}\left[e^{-\lambda \widetilde S_k}\right]\le c^k \quad\text{and}\quad \mathbb{P}_{n,\rho}(\widetilde S_k < 0)\le c^k.
\end{align}
% \item\label{part3} Let $\widetilde T_n= \inf\{k\geq 1 : S_{k}\geq B_n\}$. Then there exists some finite constant $C$ such that
% \begin{equation}
% \mathbb E_{n,\rho_0}\left[-\hat S_{\widetilde T_n}+\widetilde T_nD_\mathcal M(\rho_0\|\mathcal D)\right]\leq C
% \end{equation}
% \item\label{part4} For sufficiently small $\lambda$, there exists an $\gamma>0$ such that
% \begin{equation}\label{eq:fixed_exp_moment}
% \mathbb E_{n,\mathcal A,0}\left[e^{\gamma T_n}\right] \le Ce^{\lambda B_n}.
%     \end{equation}
\item\label{lemma19;part3}
% For every \(\sigma\in\mathcal D\) and every \(k\ge 1\),
% \begin{align}
% \Pbb_{n,\sigma}(\widetilde S_k\ge 0)
% \le e^{C_\Pi}k^{d^2-1}c^k.
% \end{align}
% where $C_\Pi$ is the same constant as in \Cref{lemma:mixture_martingale_comparison_with_pointwise}. Consequently,
There exists a finite constant  {$C_+$} satisfying
\vspace{-0.3cm}
\begin{align}
\mathbb E_{n,\sigma}\bigl[\sum_{j=1}^{T_n}\mathbf 1_{\{\widetilde S_{j-1}\ge 0\}}\bigr]\le {C_+} \quad \text{for all } \sigma\in \mathcal{D}.
\end{align}
\end{enumerate}
\end{lemma}
\vspace{-0.4cm}
\subsection{Proof of Type-I and Type-II Error Bounds}\label{sec:type_errors}
We first upper bound the Type-I error. We have
\begin{align}
    \alpha_n^\star &= \Pbb_{n,\rho}(\widetilde S_{T_n}\leq-A_n)  = \Pbb_{n,\rho}(\widetilde L_{T_n}\geq e^{A_n}) \nonumber \\
    &\leq  \frac{\mathbb E_{n,\rho}[\widetilde L_{T_n}]}{e^{A_n}} \label{eq:stopping_time_typei}\leq e^{-A_n},
\end{align} where the first inequality uses Markov's inequality, and the second follows from Doob’s optional stopping theorem \cite{durrett2019probability} applied to bounded stopping times for $\{\widetilde{L}_k\}_{k\geq 0}$, followed by Fatou's lemma. This completes the Type-I error bound.  
% that $T_n<\infty$ a.s., thus $\widetilde L_{T_n^{(m)}}\to \widetilde L_{T_n}$ a.s. as $m \to \infty$, and hence
% \[
% \mathbb E_{n,\rho}[\widetilde L_{T_n}] \leq \liminf_{m \to \infty} \mathbb E_{n,\rho}[\widetilde L_{T_n^{(m)}}] = \mathbb E_{n,\rho}[\widetilde L_0] =1
% \]

For the Type-II error, fix $\sigma \in \mathcal D$. Then, we have
\vspace{-0.2cm}
\begin{align}
\beta_n(\sigma)&=\Pbb_{n,\sigma}(d_{T_n}=0) 
\nonumber \\
% &=\textcolor{red}{\Pbb_{n,\sigma}(\widetilde S_{T_n}\geq B_n)} \nonumber \\
&=\mathbb E_{n,\rho}\biggl[\left.\frac{\dd\Pbb_{n,\sigma}}{\dd\Pbb_{n,\rho}}\right|_{\mathcal F_{T_n}}\mathbf 1 _{\{\widetilde S_{T_n}\geq B_n \}} \biggr] \label{eq:typeII_change_measure} \\[-1mm]
% &=\textcolor{red}{\mathbb E_{n,\rho}\left[ e^{- S_{T_n}(\sigma)} \mathbf 1 _{\{\widetilde S_{T_n}\geq B_n \}} \right]} \nonumber  \\
% &\leq e^{-B_n}\mathbb E_{n,\rho}\big[ e^{\widetilde S_{T_n}- S_{T_n}(\sigma)} \big] \label{eq:typeII_X-A}\\
&\leq e^{\widetilde C} e^{-B_n}\mathbb E_{n,\rho}\bigl[ T_n^{d^2-1} \bigr], \label{eq:independent_of_sigma} 
\end{align}

\vspace{-0.2cm} \hspace{-0.5cm} where \eqref{eq:typeII_change_measure} follows via a change of measure  \cite[Lemma 17 (i)]{li2022optimal}. The inequality \eqref{eq:independent_of_sigma} follows by noting that the Radon--Nikodym derivative is given by $\frac{\dd\Pbb_{n,\sigma}}{\dd\Pbb_{n,\rho}}|_{\mathcal F_{T_n}}=e^{-S_{T_n}(\sigma)}$, and since $\mathbf 1_{\{u\geq v\}}\leq e^{u-v}$, then
\vspace{-0.2cm}
\[
e^{-S_{T_n}(\sigma)}\mathbf 1 _{\{\widetilde S_{T_n}\geq B_n \}} \leq e^{-S_{T_n}(\sigma)+\widetilde S_{T_n}- B_n}.
\]
Thus \eqref{eq:independent_of_sigma} follows by applying \Cref{lemma:mixture_martingale_comparison_with_pointwise} to this bound.
% Since the above is independent of $\sigma$, we have
% \begin{equation}
%     \textcolor{red}{\beta_n^\star\leq e^{C_\Pi} e^{-B_n}\mathbb E_{n,\rho}\left[ T_n^{d^2-1} \right]}
% \end{equation}

Now suppose $T_n>k-1$. Then the stopping boundary has not been crossed by time $k-1$, so $-A_n < \widetilde S_{k-1} < B_n$, and hence $\{T_n>{k-1}\} \subseteq \{ \widetilde S_{k-1}<B_n\}$. Therefore, we have 
\vspace{-0.05cm}
\begin{align*}
&\Pbb_{n,\rho}(T_n=k)\leq \Pbb_{n,\rho}(T_n>k-1) \leq\Pbb_{n,\rho}(\widetilde S_{k-1}<B_n) \\&= \Pbb_{n,\rho}(e^{-\lambda\widetilde S_{k-1}}>e^{-\lambda B_n})
\leq \frac{\mathbb E_{n,\rho}[e^{-\lambda\widetilde S_{k-1}}]}{e^{-\lambda B_n}}\le e^{\lambda B_n} c^{k-1}
\end{align*}
where the penultimate inequality uses Markov's inequality and the last uses \Cref{lemma:lemma19}
\ref{lemma:19;part2}, with $\lambda>0$ sufficiently small and $c\in(0,1)$. Therefore, using the above we have
\vspace{-0.15cm}
%\begin{align*}
$$   \mathbb E_{n,\rho}[T_n^{d^2-1}] = \sum_{k=1}^{\infty}k^{d^2-1}\Pbb_{n,\rho}(T_n=k) \leq C'e^{\lambda B_n},
$$ where $C'<\infty$ since the polynomially weighted geometric series converges for $c\in(0,1)$. 
% Thus altogether we have 
% \begin{equation}
% \beta_n^\star \leq  Ke^{(\lambda-1) B_n} \nonumber 
% \end{equation}
% with $K=e^{C_\Pi}C<\infty$. 
Since \eqref{eq:independent_of_sigma} holds for all $\sigma \in \mathcal D$, taking the supremum yields
$
\beta_n^\star \leq  e^{\widetilde C}C'e^{(\lambda-1) B_n}.
$
Taking $n\rightarrow\infty$ and $\lambda\rightarrow0$ completes the Type-II error bound. \hspace{1em plus 1fill}\IEEEQEDhere

% % Finally,  taking $n\to\infty$ and then $\lambda\downarrow0$.
% % completes the Type-II error bound. 
% % \hspace{1em plus 1fill}\IEEEQEDhere

\subsection{Proof of Null Expectation Constraint}\label{sec:null_expectation}
Let $T_n^{(r)}:=T_n\wedge r$. By \Cref{lemma:martingales}, $\{Y_k^{(\rho)}\}_{k \geq 0}$ is a  martingale under $\Pbb_{n,\rho}$. Thus, Doob’s optional stopping theorem  yields that $\mathbb E_{n,\rho}[Y_{T_n^{(r)}}^{(\rho)}] = 0.$
We then have the following relations:\enlargethispage{-.23cm}
\vspace{-0.4cm}
\begin{align}
\mathbb E_{n,\rho}[\widetilde S_{T_n^{(r)}}] 
&=
\mathbb E_{n,\rho}\bigg[\sum_{j=1}^{\scriptscriptstyle T_n^{(r)}} \mathbb E_{n,\rho}[\widetilde Z_j \mid \mathcal F_{j-1}]\bigg] \nonumber\\[-3.75mm]
%\\
% &= \mathbb E_{n,\rho}\bigg[\sum_{j=1}^{T_n^{(m)}} \bigg(
% \mathbf 1_{\{\widetilde S_{j-1}\ge 0\}} D_{\mathcal M}(\rho\|\tilde\sigma_{j-1}) \nonumber\\
% &+ \mathbf 1_{\{\widetilde S_{j-1}<0\}} D\big(P_{\rho,m_1^\ast(\tilde\sigma_{j-1})}\|P_{\tilde\sigma_{j-1},m_1^\ast(\tilde\sigma_{j-1})}\big)
% \bigg)\bigg] \label{eq:expectation_drift_lower_bound} \\ 
&\ge D_{\mathcal M}(\rho\| \mathcal D)
\mathbb E_{n,\rho}\biggl[\sum_{j=1}^{\scriptscriptstyle T_n^{(r)}} \mathbf 1_{\{\widetilde S_{j-1}\ge 0\}} \biggr] \label{eq:expectation_drift_lower_bound} 
\\[-1mm] 
%\end{align} \begin{align}
&\geq D_{\mathcal M}(\rho\|\mathcal D)
\mathbb E_{n,\rho}[T_n^{(r)}] - O(1),
\label{eq:expectation_final_bound}
\end{align}
where in \eqref{eq:expectation_drift_lower_bound} we used \Cref{lemma:lemma19} \ref{lemma:19;part1}, combined with dropping the non-negative term corresponding to event $\{\widetilde S_{j-1}<0\}$ and that \(D_{\mathcal M}(\rho\|\tilde\sigma_{j-1}) \geq D_{\mathcal M}(\rho\|\mathcal D)\) since $\tilde\sigma_{j-1}\in \mathcal D$. Then, \eqref{eq:expectation_final_bound} follows by using $\mathbf 1_{\{\widetilde S_{j-1}\ge0\}}=1-\mathbf 1_{\{\widetilde S_{j-1}<0\}}$, together with,
% \begin{equation}\label{eq:stopping time_decomposition}
%    T_n^{(m)} = \sum_{j=1}^{T_n^{(m)}} 1 = \sum_{j=1}^{T_n^{(m)}} \mathbf{1}_{\{\widetilde S_{j-1} \ge 0\}} + \sum_{j=1}^{T_n^{(m)}} \mathbf{1}_{\{\widetilde S_{j-1}<0\}} 
% \end{equation} 
% and 
% \[
% \textcolor{red}{T_n^{(m)} = \sum_{j=1}^{T_n^{(m)}} 1 = \sum_{j=1}^{T_n^{(m)}} \mathbf{1}_{\{\widetilde S_{j-1} \ge 0\}} + \sum_{j=1}^{T_n^{(m)}} \mathbf{1}_{\{\widetilde S_{j-1} < 0\}}.}
% \]
% We then bound the expectation of the summation over the event $\{\widetilde S_{j-1}<0\}$ as follows,
\vspace{-0.25cm}
\begin{align*}
\mathbb{E}_{n,\rho}\bigl[\sum_{j=1}^{T_n^{(r)}} \mathbf{1}_{\{\widetilde S_{j-1}<0\}}\bigr]
&\hspace{-0.1cm}\le \sum_{j=1}^{\infty}\mathbb{P}_{n,\rho}(\widetilde S_{j-1}<0)\hspace{-0.1cm}\leq\sum_{j=1}^{\infty}c^{j-1}\hspace{-0.1cm}<\infty,
% &\leq \sum_{j=1}^{\infty}c^{j-1}<\infty.
\end{align*}
where the first inequality bounds the stopped sum by the infinite sum via Tonelli's theorem, and the second inequality follows from \Cref{lemma:lemma19} \ref{lemma:19;part2}, which yields a geometric series that converges since $c\in (0,1)$. Next, since there exists a constant $C>0$ such that $|\widetilde Z_k|\le C$ almost surely for all $k$ \cite[Lemma 15]{li2022optimal}, and before stopping $\widetilde S_{T_n^{(r)}-1}\leq B_n$, then $\widetilde S_{T_n^{(r)}}= \widetilde S_{T_n^{(r)}-1}+\widetilde Z_{T_n^{(r)}}
\le B_n+C $ almost surely. 
% \begin{align*}
% \widetilde S_{T_n^{(m)}}
% &\le B_n+C \qquad\text{a.s.}
% \end{align*}
% Indeed, if $T_n^{(m)}<T_n$, then $T_n^{(m)}=m$ and the test has not yet stopped, so $\widetilde S_{T_n^{(m)}}<B_n$; while if $T_n^{(m)}=T_n$, then $\widetilde S_{T_n-1}<B_n$ and $\widetilde S_{T_n}=\widetilde S_{T_n-1}+\widetilde Z_{T_n}\le B_n+C.$ 
Therefore, taking expectation and combining with \eqref{eq:expectation_final_bound} gives
\vspace{-0.1cm}
\[
\mathbb E_{n,\rho}[T_n^{(r)}]\leq \frac{B_n+O(1)}{D_{\mathcal M}(\rho\|\mathcal D)}.
\]
Finally, since $T_n^{(r)}\rightarrow T_n$ almost surely as $r\to\infty$, applying the monotone convergence theorem and taking $\limsup$ as $n\to \infty$ gives us that this test satisfies the
null side of \eqref{condition_3}. \hspace{1em plus 1fill}\IEEEQEDhere
% \begin{align*}
% \mathbb E_{n,\rho}[T_n]
% &=\lim_{m\to\infty}\mathbb E_{n,\rho}[T_n^{(m)}] \leq \frac{B_n+O(1)}{D_{\mathcal M}(\rho\|\mathcal D)} 
% \end{align*}
% Rearranging and taking $\limsup$ as $n\to \infty$ gives us that this test satisfies the null side of \eqref{condition_3}.
\vspace{-0.2cm}
\subsection{Proof of Alternative Expectation Constraints}\label{sec:alt_expectation}
Fix $\sigma \in \mathcal{D}$ and let $T_n^{(r)}:=T_n\wedge r$. By   \Cref{lemma:martingales}, the process $\{-Y_k^{(\sigma)}\}_{k \geq 0}$ is a martingale under $\mathbb{P}_{n,\sigma}$ and $\mathbb E_{n,\sigma}[Y_{T_n^{(r)}}^{(\sigma)}] = 0$. Then, we have
\vspace{-0.5cm}
\begin{align}
& -\mathbb E_{n,\sigma}[\widetilde S_{T_n^{(r)}}] 
=
\mathbb E_{n,\sigma}\bigl[\sum_{j=1}^{T_n^{(r)}} \mathbb E_{n,\sigma}[-\widetilde Z_j \mid \mathcal F_{j-1}]\bigr] \nonumber \\[-3.5mm]
% &= \mathbb E_{n,\sigma}\Bigl[\sum_{j=1}^{T_n^{(r)}} \sum_{x\in\mathcal X} P_{\sigma,M_j}(x)\log\frac{P_{\tilde\sigma_{j-1},M_j}(x)}{P_{\rho,M_j}(x)}\Bigr] \nonumber \\[-1mm]
&= \mathbb E_{n,\sigma}\Big[\sum_{j=1}^{T_n^{(r)}}\biggl( D(P_{\tilde\sigma_{j-1},M_j}\|P_{\rho,M_j})+
\nonumber \\[-3mm]
%\\& \mathbb{E}_{X \sim P_{\sigma,M_j}}[A_j(X)]- \mathbb{E}_{X \sim P_{\tilde{\sigma}_j,M_j}}[A_j(X)]\Bigr], 
&\sum_{x\in\mathcal X} \big(P_{\sigma,M_j}(x) - P_{\tilde\sigma_{j-1},M_j}(x)\big)\log\frac{P_{\tilde\sigma_{j-1},M_j}(x)}{P_{\rho,M_j}(x)}\biggr)\biggr],
\label{eq:inter_1}
\end{align}
%where $A_j(X)=\log\frac{P_{\tilde\sigma_j,M_j}(X)}{P_{\rho,M_j}(X)}$.
where the last equality follows by adding and subtracting the term $P_{\tilde\sigma_{j-1},M_j}(x)\log\frac{P_{\tilde\sigma_{j-1},M_j}(x)}{P_{\rho,M_j}(x)}$ and the definition of the classical Kullback--Leibler divergence. We now lower bound the first term and the difference term in  \eqref{eq:inter_1} separately. 

\looseness=-1
For the first term, since each summand is non-negative, inserting the indicator of $\{\widetilde S_{j-1}<0\}$ lower-bounds this sum. On this event, $D(P_{\tilde\sigma_{j-1},M_j}\|P_{\rho,M_j})\geq D_{\mathcal M}(\mathcal D\|\rho)$ because $M_j=m_1^\ast(\tilde\sigma_{j-1})$. Therefore, inserting this indicator into the sum, rewriting it as $\mathbf 1_{\{\widetilde S_{j-1}<0\}}=1-\mathbf 1_{\{\widetilde S_{j-1}\geq0\}}$, and applying \Cref{lemma:lemma19} \ref{lemma19;part3} to the resulting sum over $\mathbf 1_{\{\widetilde S_{j-1}\geq0\}}$ yields
\vspace{-0.3cm}
%\begin{small}
\begin{align}
 &\mathbb E_{n,\sigma}[\sum_{j=1}^{T_n^{(r)}} D(P_{\tilde\sigma_{j-1},M_j}\|P_{\rho,M_j})] \nonumber \\[-1mm] &\geq D_{\mathcal{M}}(\mathcal{D} \Vert \rho)    (\mathbb E_{n,\sigma}[T_n^{(r)}]-C_+). \label{eq:inter_2}
\end{align}
%\end{small}

For the difference term, taking the absolute value yields\enlargethispage{-.06cm}
\vspace{-0.2cm}
\begin{align}
&\Bigl|\sum_{x\in\mathcal X} \big(P_{\sigma,M_j}(x) - P_{\tilde\sigma_{j-1},M_j}(x)\big)\log\frac{P_{\tilde\sigma_{j-1},M_j}(x)}{P_{\rho,M_j}(x)} \Bigr| \nonumber 
\\
%\end{align} %\begin{align}
& \stackrel{(a)}{\leq} \|P_{\sigma,M_j} - P_{\tilde\sigma_{j-1},M_j}\|_1  \bigg\| \log\frac{P_{\tilde\sigma_{j-1},M_j}}{P_{\rho,M_j}}\bigg\|_\infty \nonumber\\[-0.1cm]
&\stackrel{(b)}{\leq} C\|P_{\sigma,M_j} - P_{\tilde\sigma_{j-1},M_j}\|_1 \stackrel{(c)}{\leq} C \sqrt{2 D(P_{\sigma,M_j} \Vert P_{\tilde\sigma_{j-1},M_j})}, 
\nonumber
%\label{eq:p_and_cs}
\end{align}
\looseness=-1
    where $(a)$  follows from H\"older's inequality, $(b)$ follows since $\bigl|\log\tfrac{P_{\tilde\sigma_{j-1},M}(x)}{P_{\rho,M}(x)}\bigr|\leq C$ almost surely for all $x\in \mathcal X$ for some $C>0$ \cite[Lemma 15]{li2022optimal}, and $(c)$ follows from Pinsker's inequality. Using the above in \eqref{eq:inter_1}, then applying Cauchy--Schwarz twice, first to pull the square root outside the stopped sum and then under the expectation to separate the resulting factors, yields
%Next, substituting \eqref{eq:p_and_cs} into the stopped sum and applying Pinsker’s inequality term-wise followed by Cauchy--Schwarz over the stopped sum and then Cauchy--Schwarz under expectation gives, \vspace{-0.3cm}
$
\mathbb E_{n,\sigma}\bigl[\sum_{j=1}^{T_n^{(r)}}|\mbox {diff. term}|\bigr]\leq C\sqrt{2\mathbb E_{n,\sigma}[T_n^{(r)}]\mathbb E_{n,\sigma}\biggl[\sum_{j=1}^{T_n^{(r)}}D(P_{\sigma,M_j}\|P_{\tilde\sigma_{j-1},M_j})\biggr]}.$ Since, $\mathbb E_{n,\sigma}[\sum_{j=1}^{T_n^{(r)}}D(P_{\sigma,M_j}\|P_{\tilde\sigma_{j-1},M_j})]=\mathbb E_{n,\sigma}[\widetilde S_{T_n^{(r)}} - S_{T_n^{(r)}}(\sigma)]$, then applying \Cref{lemma:mixture_martingale_comparison_with_pointwise} and using Jensen's inequality to pull the $\log$ out of the expectation gives 
\vspace{-0.15cm}
\begin{align*}
    \mathbb E_{n,\sigma}[\sum_{j=1}^{T_n^{(r)}}D(P_{\sigma,M_j}\|P_{\tilde\sigma_{j-1},M_j})]
%&\leq \mathbb E_{n,\sigma}\left[\widetilde C + (d^2-1)\log(T_n^{(r)})\right]\\
&\leq \hspace{-0.1cm}\widetilde{C} + \hspace{-0.1cm}(d^2-1)\log\Ebb_{n,\sigma}[T_n^{(r)}].
\end{align*}
%where we used Jensen's inequality. 
Since $|a|\leq b$ implies $-b\leq a$, the above and \eqref{eq:inter_2} yield
\vspace{-0.2cm}
\begin{align}
&\hspace{-0.22cm}-\mathbb E_{n,\sigma}[\widetilde S_{T_n^{(r)}}]\geq D_{\mathcal M}(\mathcal D\|\rho)\left(\mathbb E_{n,\sigma}[T_n^{(r)}]-C_+\right) \label{eq:final_type_ii}\\
&-C\sqrt{2\mathbb E_{n,\sigma}[T_n^{(r)}]\left(\widetilde C+ (d^2-1)\log\mathbb E_{n,\sigma}[T_n^{(r)}]\right)}. \nonumber
% \label{eq:alt_lower_bound_final}
\end{align} 
\looseness=-1
Note, by the same argument as before, $-\widetilde S_{T_n^{(r)}}\le A_n+C$ almost surely. Since the second term of \eqref{eq:final_type_ii} is sub-linear in $\mathbb E_{n,\sigma}[T_n^{(r)}]$, there is a finite constant $\Gamma_{\tau/2}$, with $\tau$ as in \Cref{prop:achievability}, such that once $\mathbb E_{n,\sigma}[T_n^{(r)}]\ge \Gamma_{\tau/2}$, this sub-linear term is at most $\tfrac{\tau}{2}\mathbb E_{n,\sigma}[T_n^{(r)}]$. In this case, rearranging \eqref{eq:final_type_ii} gives $(D_{\mathcal M}(\mathcal D\|\rho)-\tfrac{\tau}{2})\mathbb E_{n,\sigma}[T_n^{(r)}]\le A_n+O(1)$; if instead $\mathbb E_{n,\sigma}[T_n^{(r)}]< \Gamma_{\tau/2}$, then this holds trivially as $\mathbb E_{n,\sigma}[T_n^{(r)}]$ is $O(1)$. Thus, using the monotone convergence theorem yields
\vspace{0cm}
\begin{align*}
\mathbb E_{n,\sigma}[T_n] \leq \frac{A_n+O(1)}{D_{\mathcal M}(\mathcal D\|\rho)-\frac{\tau}{2}}=n\left(\frac{D_{\mathcal M}(\mathcal D\|\rho)-\tau}{D_{\mathcal M}(\mathcal D\|\rho)-\frac{\tau}{2}}\right) +O(1).
\end{align*}
\looseness=-2
Finally, since $\tau<D_{\mathcal M}(\mathcal D\|\rho)$, we have $\xi:=\frac{D_{\mathcal M}(\mathcal D\|\rho)-\tau}{D_{\mathcal M}(\mathcal D\|\rho)-\tau/2}<1$. Moreover, the bound above is uniform over $\sigma\in\mathcal D$, and therefore $\sup_{\sigma\in\mathcal D}\mathbb E_{n,\sigma}[T_n]-n\leq(\xi-1)n+O(1)$. Taking $\limsup_{n\to\infty}$, we have $(\xi-1)n\to-\infty$, which shows that this test satisfies the alternative side of \eqref{condition_3}.\hspace{1em plus 1fill}\IEEEQEDhere
\section{Conclusion}\enlargethispage{-.04cm}
We characterized the optimal error exponent region for adaptive composite SQHT under an expected sample size constraint, with the exponents characterized by the minimal measured relative entropies between the null state and the alternative set. Whether the stronger pair $\bigl(\inf_{\sigma\in\mathcal{D}} D(\sigma\|\rho), \inf_{\sigma\in\mathcal{D}} D(\rho\|\sigma)\bigr)$ is achievable via block measurement schemes remains open. Natural extensions include composite--composite and multi-hypothesis settings, as well as alternative sample size constraints.
%This work characterized the optimal error-exponent region for adaptive SQHT in the simple-null versus composite-alternative setting under an expected sample size constraint. The optimal Type-I and worst-case Type-II error exponents are given by the minimal measured relative entropies between the null state and the alternative set. An interesting open question is whether the stronger exponent pair $(\inf_{\sigma \in \mathcal{D}}D(\sigma \Vert \rho), \inf_{\sigma \in \mathcal{D}}D(\rho \Vert \sigma))$ is achievable, potentially via block measurement schemes. Future work may also extend the analysis to composite–composite and multi-hypothesis settings, as well as alternative sample size constraints.
\bibliographystyle{IEEEtran}
\bibliography{references}
\newpage
\appendix
\crefalias{subsection}{appendix} %to label appendix sections as "Appendix A" not "Section A"
Here we present detailed proofs of the technical lemmas from \Cref{sec:prelim}. Specifically, \Cref{appendix:prem} collects the technical ingredients used throughout the Appendix, and \Cref{appendix:proof_of_lemma_8,appendix:proof_of_lemma_9,appendix:proof_of_lemma_10} prove \Cref{lemma:martingales,lemma:mixture_martingale_comparison_with_pointwise,lemma:lemma19}, respectively.
\subsection{Preliminaries}\label{appendix:prem}
We begin by fixing the notation used in the Appendix. Let $\mathsf{H}_d:=\{H\in \mathbb C^{d\times d}: H=H^\dagger\}$ denote the real vector space of Hermitian matrices in $\mathbb C^{d \times d}$, which has real dimension $d^2$. In what follows, $\mathsf H_d$ is equipped with the Hilbert--Schmidt inner product $\langle A, B \rangle := \tr[AB]$. Then, the affine space $\mathsf A_d$ can  be written equivalently as $\mathsf A_d=\{\rho \in \mathsf H_d: \tr[\rho]=1\}$. Since $\tr[\cdot]$ is a non-zero real linear functional on $\mathsf H_d$, $\mathsf A_d$ is an affine hyperplane of $\mathsf H_d$, and hence has affine dimension $d^2-1$. Moreover, the associated translation space of $\mathsf A_d$ is the real vector space $\mathsf V_d:=\{H\in\mathsf H_d : \tr[H]=0\}$, which is a linear subspace of $\mathsf H_d$ with real dimension $d^2-1$. Consequently, $\mathsf A_d =  \tfrac{\mathbb I}{d} + \mathsf V_d$, where $\tfrac{\mathbb I}{d}$ is the maximally mixed state. We equip $\mathsf V_d$ with the Lebesgue measure induced by the Hilbert–Schmidt inner product, denoted by $\Vol_{\mathsf V_d}$.  We also recall that the set of density matrices, $\mathsf S_d:=\{\rho\in \mathsf A_d:\rho\succeq0\}$, is a convex subset of $\mathsf A_d$ with affine dimension $d^2-1$. \Cref{fig:spaces} provides a geometric illustration of the spaces
$\mathsf H_d$, $\mathsf A_d$, $\mathsf V_d$, and $\mathsf S_d$ in the eigenvalue
plane for $d=2$, to aid intuition.
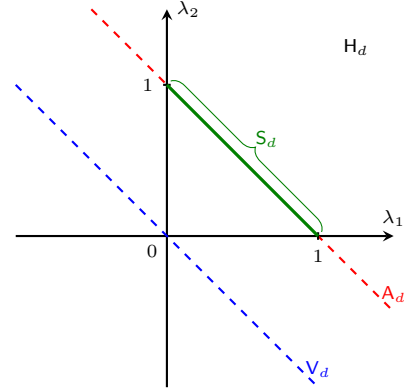
\begin{figure}[h!]
    \centering
\begin{tikzpicture}[decoration={brace,amplitude=5pt}] 
        \draw[thick,-stealth](-2,0)--(3,0)node[above]{$\lambda_1$};
        \draw[thick,-stealth](0,-2)--(0,3)node[right]{$\lambda_2$};
        \node[below left]{$0$};
        \draw[thick,red,dashed](-1,3)--(3,-1)node[above]{$\mathsf A_d$};
        \draw[decorate, darkgreen](.05,2.05)--(2.05,.05)node[midway,above right]{$\mathsf S_d$};
        \draw[thick](2,.05)--(2,-.05)node[below]{$1$};
        \draw[thick](.05,2)--(-.05,2)node[left]{$1$};
        \draw[very thick,darkgreen](0,2)--(2,0);
        \draw[thick,blue,dashed](-2,2)--(2,-2)node[above]{$\mathsf V_d$};
        \node at(2.5,2.5){$\mathsf H_d$};
    \end{tikzpicture}
    \caption{Illustration in the eigenvalue plane of $\mathsf H_d$ for $d=2$. The line $\mathsf A_d$ represents the trace-one constraint $\lambda_1+\lambda_2=1$, while $\mathsf S_d$ is the line segment where $\lambda_1,\lambda_2\ge0$ and $\lambda_1+\lambda_2=1$. The line $\mathsf V_d$ represents the associated translation space, corresponding to the trace-zero constraint $\lambda_1+\lambda_2=0$.}
    \label{fig:spaces}
\end{figure}

Note that since all spaces considered here are finite-dimensional, all norms induce the same topology; when a specific norm is needed, we use the Euclidean norm $\|\cdot\|_2$ on $\mathbb R^{d^2-1}$ and the trace-norm $\|\cdot\|_1$ on all previously defined matrix spaces.
We further write $\interior_{\mathbb R^{d^2-1}}(\cdot)$ and $\interior_{\mathsf A_d}(\cdot)$ for the interior relative to $\mathbb R^{d^2-1}$ and $\mathsf A_d$, respectively. Explicitly, for $S\subset\mathbb R^{d^2-1}$,\vspace{-0.4cm}
\begin{multline}
\interior_{\mathbb R^{d^2-1}}(S):=\Bigl\{x\in S:  \text{ there exists an }\varepsilon>0 \text{ such } \\ \text{that } \{x'\in \mathbb R^{d^2-1}: \|x'-x\|_2<\varepsilon\}\subseteq S \Bigr\},
\end{multline}
\newpage 
\hspace{-0.32cm}and similarly for $\mathcal D \subset \mathsf A_d$, \vspace{-0.2cm}
\begin{multline}
\interior_{\mathsf A_d}(\mathcal D):=\Bigl\{\sigma\in \mathcal D:  \text{ there exists }\varepsilon>0 \text{ such } \\ \text{that } \{\sigma'\in \mathsf A_d: \|\sigma'-\sigma\|_1<\varepsilon\}\subseteq \mathcal D \Bigr\}.
\end{multline}
Correspondingly, we write $\Vol_{\mathbb  R^{d^2-1}}$ as the Lebesgue measure on $\mathbb R^{d^2-1}$ and $\Vol_{\mathsf A_d}$ as the Lebesgue measure on $\mathsf A_d$, defined via the translation $\Vol_{\mathsf A_d}(A):=\Vol_{\mathsf V_d}(A-\tfrac{\mathbb I}{d})$ for all Borel measurable sets $A\subseteq\mathsf A_d$. Lastly, $\Pi$ denotes the probability measure on $\mathcal D$ obtained by normalizing the restriction of $\Vol_{\mathsf{A}_d}$ to $\mathcal D$, i.e., for Borel measurable $A \subseteq \mathcal D$,
\begin{equation}
  \Pi(A):=\frac{\Vol_{\mathsf{A}_d}(A)}{\Vol_{\mathsf{A}_d}(\mathcal D)}.
\end{equation}
Since $\interior_{\mathsf{A}_d}(\mathcal D)\ne \emptyset$ and  $\mathcal D$ is a Borel measurable subset of $\mathsf A_d$, as $\mathcal D$ is compact, it follows that $\Vol_{\mathsf{A}_d}(\mathcal D)>0$ and hence $\Pi$ is well defined. Furthermore, for $x\in \mathbb R^{d^2-1}$ and $\varepsilon>0$, we define the closed Euclidean $\varepsilon$-ball about $x$ as
\[
\mathcal B_{\Euc}(x,\varepsilon):=\{x'\in \mathbb R^{d^2-1} : \|x'-x\|_2\le \varepsilon\},
\]
and for $\sigma \in \mathcal D$, the closed trace-norm $\varepsilon$-ball in $\mathcal D$ about $\sigma$ is
\[
\mathcal B_1(\sigma,\varepsilon) := \{\sigma' \in \mathcal D : \|\sigma' - \sigma\|_1 \le \varepsilon\}.
\]

We now justify that the maximizing PVMs can be selected measurably in $\sigma$, ensuring that $M_k$ is $\mathcal F_{k-1}$-measurable. Let $\mathcal X=\{1,\ldots,d\}$ and let $\mathcal P_d\subseteq\mathcal M_{\mathcal X}$ denote the set of ordered rank-one PVMs on $\mathbb C^d$, which is a compact metrizable space, being a closed subset of the compact set $\mathcal M_{\mathcal X}$. Define
\[
f_0(\sigma,m):=D(P_{\rho,m}\|P_{\sigma,m}) \text{ and } f_1(\sigma,m):=D(P_{\sigma,m}\|P_{\rho,m}),
\]
so that, by \cite[Theorem~2]{Berta2017}, $D_{\mathcal M}(\rho\|\sigma) =\max_{m\in\mathcal P_d}f_0(\sigma,m)$ and $D_{\mathcal M}(\sigma\|\rho) =\max_{m\in\mathcal P_d}f_1(\sigma,m)$ for every $\sigma\in\mathcal D$. Since $\rho$ and every $\sigma\in\mathcal D$ are full-rank and each $m(x)$ is a rank-one projection, we have $\tr[m(x)\rho]>0$ and $\tr[m(x)\sigma]>0$ for all $x\in\mathcal X$. Hence, for $i=0,1$, $f_i$ is continuous on $\mathcal D\times\mathcal P_d$, and therefore a Carath\'eodory function. Since the constant correspondence $\sigma\mapsto\mathcal P_d$ is weakly measurable with non-empty compact values, the measurable maximum theorem \cite[Theorem~18.19]{aliprantis2006infinite} yields, for each $i\in\{0,1\}$, a Borel-measurable selection $m_i^\ast:\mathcal D\to\mathcal P_d$ of $\sigma\mapsto\arg\max_{m\in\mathcal P_d}f_i(\sigma,m)$. As $\tilde\sigma_{k-1}$ is $\mathcal F_{k-1}$-measurable and $\{\widetilde S_{k-1}\ge0\}\in\mathcal F_{k-1}$, the piecewise-defined POVM in \eqref{eq:mu_k} is therefore $\mathcal F_{k-1}$-measurable.

% Recall that, for each $\sigma\in\mathcal D$, the POVMs used in the
% mixture-SQPRT are chosen as
% \begin{align*}
%  m^{\ast}_0(\sigma)&\in \arg \sup_{\Xcal} \sup_{m \in \mathcal{M}_{\Xcal}}   D\left(P_{\rho, m} \| P_{\sigma, m}\right), \\
%    m^{\ast  }_1(\sigma)&\in \arg \sup_{\Xcal} \sup_{m \in \mathcal{M}_{\Xcal}}   D\left(P_{\sigma, m} \| P_{\rho, m}\right).
% \end{align*} 

% We now justify that these maximizers exist and can be chosen measurably.
% Fix $\sigma\in\mathcal D$. Since $m \mapsto D(P_{\rho,m}\|P_{\sigma,m})$ and $m \mapsto D(P_{\sigma,m}\|P_{\rho,m})$  are continuous and $\mathcal M_{\mathcal X}$ is compact, the extreme value theorem implies that the corresponding maxima are attained. Since the maximizer need not be unique, we fix one measurable choice of maximizer for each $\sigma\in\mathcal D$. Such choices exist by the measurable maximum theorem, since $\mathcal D$ and $\mathcal M_{\mathcal X}$ are compact metric spaces  under their induced finite-dimensional metrics, and the objective functions $(\sigma,m)\mapsto D(P_{\rho,m}\|P_{\sigma,m})$ and $(\sigma,m)\mapsto D(P_{\sigma,m}\|P_{\rho,m})$ are continuous on $\mathcal D\times\mathcal M_{\mathcal X}$.

The remainder of this section presents some auxiliary results, which will be used in the subsequent sections. First, we show that the mixture-SQLLR admits the incremental representation as in \eqref{eq:incremental}.
\begin{lemma}\label{lemma:bary}
The mixture-SQLLR admits the following incremental representation:
   \begin{align*}
    \widetilde{S}_k=\sum_{j=1}^k \widetilde{Z}_j, \quad \widetilde Z_j := \log \frac{\tr\big(M_j(X_j)\rho\big)}{\tr\big(M_j(X_j)\tilde\sigma_{j-1}\big)}.
\end{align*} 
\end{lemma} 
\begin{IEEEproof}
Indeed, by the definition of $\widetilde L_k$,
\begin{align}
\widetilde L_k 
&=  \int_{\mathcal D} e^{-S_{k}(\sigma)} \Pi(\dd\sigma) \nonumber \\
% &=  \int_{\mathcal D} e^{-S_{k-1}(\sigma)-Z_{k}(\sigma)} \Pi(\dd\sigma)\\
% &= \int_{\mathcal D} e^{-S_{k-1}(\sigma)-\log \frac{\tr[M_k(X_k)\rho]}{\tr[M_k(X_k)\sigma]}} \Pi(\dd\sigma) \nonumber \\
% &= \int_{\mathcal D} e^{-S_{k-1}(\sigma)+\log \frac{\tr[M_k(X_k)\sigma]}{\tr[M_k(X_k)\rho]}} \Pi(\dd\sigma) \\
&= \int_{\mathcal D} e^{-S_{k-1}(\sigma)}\frac{\tr[M_k(X_k)\sigma]}{\tr[M_k(X_k)\rho]}\Pi(\dd\sigma) \nonumber \\
% &= \frac{1}{\tr[M_k(X_k)\rho]}\int_{\mathcal D}\tr[M_k(X_k)\sigma] e^{-S_{k-1}(\sigma)} \Pi(\dd\sigma).
&= \frac{1}{\tr[M_k(X_k)\rho]}\int_{\mathcal D}\tr[M_k(X_k)\sigma]\widetilde L_{k-1}\widetilde\Pi_{k-1}(\dd\sigma) \label{eq:def_of_pi} \\
&=\frac{\widetilde L_{k-1}\tr[M_k(X_k)\tilde\sigma_{k-1}]}{\tr[M_k(X_k)\rho]} \label{eq:linear_trace}
\end{align}
where in \eqref{eq:def_of_pi} we used that 
by the definition of $\widetilde\Pi_{k-1}$,
\[
e^{-S_{k-1}(\sigma)}\Pi(\dd\sigma)
=
\widetilde L_{k-1}\widetilde\Pi_{k-1}(\dd\sigma),
\] 
and in \eqref{eq:linear_trace} we used the linearity of the trace.
% we have 
% \[
% \widetilde\Pi_{k-1}(\dd\sigma)
% = \frac{L_{k-1}(\sigma)\Pi(\dd\sigma)}{\widetilde L_{k-1}} =
% \frac{e^{-S_{k-1}(\sigma)}\Pi(\dd\sigma)}
% {\widetilde L_{k-1}},
% \]
% and hence
% \[
% e^{-S_{k-1}(\sigma)}\Pi(\dd\sigma)
% =
% \widetilde L_{k-1}\widetilde\Pi_{k-1}(\dd\sigma).
% \]
% Therefore,
% \begin{align}
% \int_{\mathcal D}\tr[M_k(X_k)\sigma]e^{-S_{k-1}(\sigma)}\Pi(\dd\sigma) 
% &=
% \int_{\mathcal D}\tr[M_k(X_k)\sigma]\widetilde L_{k-1}\widetilde\Pi_{k-1}(\dd\sigma) \nonumber \\
% &=
% \widetilde L_{k-1}\int_{\mathcal D}\tr[M_k(X_k)\sigma]\widetilde\Pi_{k-1}(\dd\sigma) \nonumber \\
% &=
% \widetilde L_{k-1}\tr\left[M_k(X_k)\int_{\mathcal D}\sigma\widetilde\Pi_{k-1}(\dd\sigma)\right] \nonumber  \\
% &=
% \widetilde L_{k-1}\tr[M_k(X_k)\tilde\sigma_{k-1}], \nonumber 
% \end{align}
% Thus
% \[
% \widetilde L_{k}= \frac{\widetilde L_{k-1}\tr[M_k(X_k)\tilde\sigma_{k-1}]}{\tr[M_k(X_k)\rho]}.
% \]
Iterating this recursion after $k$ steps obtains 
\[
\widetilde L_k= \widetilde L_{0}\prod_{j=1}^k\frac{\tr[M_j(X_j)\tilde\sigma_{j-1}]}{\tr[M_j(X_j)\rho]}.
\]
Since $\widetilde L_{0}=e^{-0}=1$, taking logarithms and multiplying by $-1$ yields the desired result. \end{IEEEproof}
% \hspace{1em plus 1fill}\IEEEQEDhere
% \[
% \widetilde S_k
% = \sum_{j=1}^k \log \frac{\tr[M_j(X_j)\rho]}{\tr[M_j(X_j)\tilde\sigma_{j-1}]}.
% \]
% \IEEEQEDhere
% \textcolor{blue}{We first show that $e^{S_k(\sigma)}$ is the Radon--Nikodym derivative.
% \begin{IEEEproof}[Proof that $e^{S_k(\sigma)}$ is the Radon--Nikodym derivative]
% Fix $\sigma\in\mathcal D$. By the chain rule for conditional probabilities and the fact that the measurement selection rule $\{\mu_j\}$ is the same under both measures, we conclude that
% \begin{align*}
% \left.\frac{\dd\Pbb_{n,\rho}}{\dd\Pbb_{n,\sigma}}\right|_{\mathcal F_k}
% &=\prod_{j=1}^k \frac{\mu_j(\cdot\vert x_1^{j-1},m_1^{j-1})\tr[M_j(X_j)\rho]}{\mu_j(\cdot\vert x_1^{j-1},m_1^{j-1})\tr[M_j(X_j)\sigma]}\\
% &= \prod_{j=1}^k \frac{\tr[M_j(X_j)\rho]}{\tr[M_j(X_j)\sigma]}
% = e^{S_k(\sigma)}.
% \end{align*}
% Thus $e^{S_k(\sigma)}$ is the Radon--Nikodym derivative of $\Pbb_{n,\rho}$ with respect to $\Pbb_{n,\sigma}$.
% \end{IEEEproof}}

% In the remainder of this section, we establish several auxiliary results used in the proof of \Cref{lemma:mixture_martingale_comparison_with_pointwise}.

% We now establish 
% % a relationship  between the norms on the finite-dimensional normed spaces $(\mathsf V_d,\|\cdot\|_1)$ and $(\mathbb{R}^{d^2-1},\|\cdot\|_2)$ via 
% the following standard result. This will be used in the proof of \Cref{lemma:balls}, and in turn will help prove \Cref{lemma:mixture_martingale_comparison_with_pointwise}.
We now show that, for any fixed $\sigma\in\mathcal D$, $e^{S_k(\sigma)}$ is the Radon--Nikodym derivative of $\Pbb_{n,\rho}$ with respect to $\Pbb_{n,\sigma}$ restricted to $\mathcal{F}_k$, i.e.
\begin{equation}
   \left.\frac{\dd\Pbb_{n,\rho}}{\dd\Pbb_{n,\sigma}}\right|_{\mathcal F_k} = e^{S_k(\sigma)}.
\end{equation}
\begin{IEEEproof}
Indeed, expanding the definitions of $\Pbb_{n,\rho}$ and $\Pbb_{n,\sigma}$ we obtain
\begin{align*}
    \left.\frac{\dd\Pbb_{n,\rho}}{\dd\Pbb_{n,\sigma}}\right|_{\mathcal F_k}&=\frac{\prod_{j=1}^k\mu_j(m_j|x_1^{j-1},m_1^{j-1})\tr[m_j(x_j)\rho]}{\prod_{j=1}^k\mu_j(m_j|x_1^{j-1},m_1^{j-1})\tr[m_j(x_j)\sigma]}\\
    &=\prod_{j=1}^k\frac{\tr[m_j(x_j)\rho]}{\tr[m_j(x_j)\sigma]}
    % &=e^{\sum_{j=1}^k\log\frac{\tr[m_j(x_j)\rho]}{\tr[m_j(x_j)\sigma]}}
    = e^{S_k(\sigma)}&&\IEEEQEDhere
\end{align*} 
\end{IEEEproof}

To conclude this section, we state a well-known finite-dimensional norm-equivalence result, which will be used to prove  \Cref{lemma:mixture_martingale_comparison_with_pointwise}. The proof is included for completeness.
\begin{lemma}\label{lemma:isomorphism}
Let $J:\mathsf V_d \to \mathbb{R}^{d^2-1}$ be a linear isomorphism between the finite-dimensional normed spaces $(\mathsf V_d,\|\cdot\|_1)$ and $(\mathbb{R}^{d^2-1},\|\cdot\|_2)$. Then, there exists a constant $a>0$ such that for all $H\in\mathsf V_d$,
\[
\|JH\|_2 \ge a\|H\|_1.
\]
\end{lemma}
\begin{IEEEproof}
The statement holds trivially for $H=0$, so it suffices to consider $H\neq 0$. Since $J$ is a bounded linear operator, it is continuous and thus the map $H \mapsto \|JH\|_2$ is continuous as it is the composition of $H \mapsto JH$ and the Euclidean norm. Consider the unit sphere $S:=\{H\in\mathsf V_d:\|H\|_1=1\}$, which is compact in finite dimensions. Because the function $H \mapsto \|JH\|_2$ is continuous on the compact sphere $S$, it attains its minimum on $S$, say $a:=\min_{\|H\|_1=1}\|JH\|_2$. As $J$ is injective and $H\ne0$, it follows that $\|JH\|_2\neq 0$ for all $H\in S$, and hence $a>0$.  Define $U:=H/\|H\|_1$. Then $\|U\|_1=1$, so $U\in S$ and hence
\[
\|JH\|_2=\|JU\|_2\|H\|_1 \ge a\|H\|_1.
\] 
This completes the proof.
\end{IEEEproof}
\subsection{Proof of \Cref{lemma:martingales}}\label{appendix:proof_of_lemma_8}
In this section we prove that each of the processes $\{L_k(\sigma)\}_{k\ge0}$, $\{\widetilde L_k\}_{k\ge0}$ and $\{Y_k^{(\rho)}\}_{k \geq 0}$ are martingales.
% \textcolor{blue}{For properties on martingales and filtration, see also \cite[Section 4]{li2022optimal} and the references therein.}
\begin{IEEEproof}[Proof of $\{L_k(\sigma)\}_{k\ge0}$]
For fixed $\sigma\in\mathcal D$, the process $\{L_k(\sigma)\}_{k\ge0}$ is adapted to $\{\mathcal{F}_k\}_{k\geq 0}$ as $L_k(\sigma)=e^{-S_k(\sigma)}$ and $S_k(\sigma)$ is $\mathcal F_k$-measurable. Moreover, $L_k(\sigma)\ge0$, so the conditional expectation is well-defined and therefore,
\begin{align*}
\mathbb E_{n,\rho}[L_k(\sigma) \mid \mathcal F_{k-1}] &= e^{-S_{k-1}(\sigma)}\mathbb E_{n,\rho}\Bigl[\tfrac{\tr[M_k(X_k)\sigma]}{\tr[M_k(X_k)\rho]}\Bigm| \mathcal F_{k-1}\Bigr] \\
% &= e^{-S_{k-1}(\sigma)} \sum_{x}\Pbb_{n,\rho} \left(X_k=x\mid\mathcal F_{k-1} \right)e^{-Z_k(\sigma; x)}\\
&=  e^{-S_{k-1}(\sigma)} \sum_{x} \tr[M_k(x)\rho] \frac{\tr[M_k(x)\sigma]}{\tr[M_k(x)\rho]}\\
&=  e^{-S_{k-1}(\sigma)} \sum_{x}\tr[M_k(x)\sigma]\\
% &= e^{-S_{k-1}(\sigma)} \tr\left[ \sum_{x}M_k(x)\sigma\right]\\
% &= e^{-S_{k-1}(\sigma)} \tr[\sigma] \\
&= e^{-S_{k-1}(\sigma)}\\
&= L_{k-1}(\sigma).
\end{align*} Taking expectations and using $L_0(\sigma)=1$, we obtain by induction that $\mathbb E_{n,\rho}[L_k(\sigma)]=1$ for all $k$, and hence $L_k(\sigma)$ is $L^1$-integrable, that is $L_k(\sigma)\in L^1(\Pbb_{n,\rho})$. Therefore $\{L_k(\sigma)\}_{k\ge0}$ is a martingale under $\Pbb_{n,\rho}$ with respect to $\{\mathcal{F}_k\}_{k\geq 0}$. 
\end{IEEEproof}
\begin{IEEEproof}[Proof of $\{\widetilde L_k\}_{k\ge0}$] 
For each fixed $\sigma\in\mathcal D$, the random variable $L_k(\sigma)$ is given pointwise by
\[
L_k(\omega,\sigma):=\prod_{j=1}^k
\frac{\tr\big(M_j(X_j(\omega))\sigma\big)}
{\tr\big(M_j(X_j(\omega))\rho\big)}
\]
where $\omega\in\Omega$ denotes a realization of the observation history $(X_1,X_2,\dots)$, and $X_j(\omega)$ is the $j$-th element of $\omega$. Similarly, define
\[
\widetilde L_k(\omega):=\int_{\mathcal D} L_k(\omega,\sigma)\Pi(\dd\sigma).
\]
Since each $M_j$ is $\mathcal F_{j-1}$-measurable and $\tr\big(M_j(X_j(\omega))\sigma\big)$ is continuous in $\sigma$,  it follows that $(\omega,\sigma)\mapsto L_k(\omega,\sigma)$ is jointly measurable. Moreover, for each fixed $\sigma\in\mathcal D$, $L_k(\sigma)$ is a non-negative $\mathcal F_k$-measurable random variable and therefore we get the chain of equalities
\begin{align}
\mathbb E_{n,\rho}\left[\widetilde L_k \mid \mathcal F_{k-1}\right] \nonumber
&=
\mathbb E_{n,\rho}\left[\int_{\mathcal D} L_k(\sigma)\Pi(\dd\sigma) \Bigm| \mathcal F_{k-1} \right] \nonumber\\
&=
\int_{\mathcal D} \mathbb E_{n,\rho}\left[L_k(\sigma)\mid \mathcal F_{k-1}\right] \Pi(\dd\sigma) \label{eq:conditional Fubini theorem}\\
% &= \int_{\mathcal D} L_{k-1}(\sigma) \mathbb E_{n,\rho}\biggl[ \frac{\tr[M_k(X_k)\sigma]}{\tr[M_k(X_k)\rho]} \Bigm| \mathcal F_{k-1} \biggr]\Pi(\dd\sigma) \nonumber\\
% &= \int_{\mathcal D} L_{k-1}(\sigma) \sum_{x} \Pbb_{n,\rho}(X_k=x\mid\mathcal F_{k-1}) \frac{\tr[M_k(x)\sigma]}{\tr[M_k(x)\rho]} \Pi(\dd\sigma) \nonumber\\
% &= \int_{\mathcal D} L_{k-1}(\sigma)
% \sum_{x} \tr[M_k(x)\rho] \frac{\tr[M_k(x)\sigma]}{\tr[M_k(x)\rho]} \,\Pi(\dd\sigma) \nonumber\\
% &= \int_{\mathcal D} L_{k-1}(\sigma)\sum_{x} \tr[M_k(x)\sigma] \,\Pi(\dd\sigma) \nonumber\\
&= \int_{\mathcal D} L_{k-1}(\sigma)\Pi(\dd\sigma)= \widetilde L_{k-1}. \nonumber
\end{align}
 where in \eqref{eq:conditional Fubini theorem} we use that $(\omega,\sigma)\mapsto L_k(\omega,\sigma)$ is jointly measurable and hence the conditional Fubini theorem \cite[Theorem 27.17]{schilling2017measures} applies.
 Moreover, since $\widetilde L_k \ge 0$, we may apply Tonelli's theorem \cite[Theorem 5.28]{axler2020measure} to obtain
 \begin{align*}
  \mathbb E_{n,\rho}[\widetilde L_k] &= \mathbb E_{n,\rho}\left[\int_{\mathcal D} L_k(\sigma)\Pi(\dd\sigma)\right] \\&= \int_{\mathcal D} \mathbb E_{n,\rho}[L_k(\sigma)]\Pi(\dd\sigma)
\\&= \int_{\mathcal D} 1\Pi(\dd\sigma) = 1,  
 \end{align*}
where we used that $\mathbb E_{n,\rho}[L_k(\sigma)]=1$ for all $\sigma\in\mathcal D$. Hence $\widetilde L_k \in L^1(\Pbb_{n,\rho})$, and therefore $\{\widetilde L_k\}_{k\ge0}$ is a martingale under $\Pbb_{n,\rho}$ with respect to $\{\mathcal{F}_k\}_{k\geq 0}$.
\end{IEEEproof}
\begin{IEEEproof}[Proof of  $\{Y_k^{(\rho)}\}_{k \geq 0}$]
For $k=1$, since $\widetilde S_1=\widetilde Z_1$, we have
\begin{align*}
\mathbb E_{n,\rho}[Y_1^{(\rho)}\mid\mathcal F_0]&=-\mathbb E_{n,\rho}[\widetilde Z_1\mid\mathcal F_0]
+
\mathbb E_{n,\rho}[\widetilde Z_1\mid\mathcal F_0]\\
&= 0= Y_0^{(\rho)}
\end{align*}
For $k\geq2$ we have the following chain of equalities
\begin{align*}
\mathbb E_{n,\rho}[Y_k^{(\rho)} \mid \mathcal F_{k-1}]
&= -\widetilde S_{k-1} - \mathbb E_{n,\rho}\biggl[\widetilde Z_k \mid \mathcal F_{k-1}\biggr]   \\
&+\mathbb E_{n,\rho}\biggl[\sum_{j=1}^k \mathbb E_{n,\rho}\bigl[\widetilde Z_j \mid \mathcal F_{j-1}\bigr] \mid \mathcal F_{k-1}\biggr] \\
&= -\widetilde S_{k-1} - \mathbb E_{n,\rho}\left[\widetilde Z_k \mid \mathcal F_{k-1}\right] \\
&+ \mathbb E_{n,\rho}\biggl[\widetilde Z_k \mid \mathcal F_{k-1}\biggr] + \sum_{j=1}^{k-1} \mathbb E_{n,\rho}\biggl[\widetilde Z_j \mid \mathcal F_{j-1}\biggr] \\
&= Y_{k-1}^{(\rho)},
\end{align*}
where we used that since for every $j\leq k-1$, $\mathbb E_{n,\rho}[\widetilde Z_j \mid \mathcal F_{j-1}]$ is $\mathcal F_{j-1}$-measurable, and hence $\mathcal F_{k-1}$-measurable since $\mathcal F_{j-1}\subseteq\mathcal F_{k-1}$. Therefore,  $\sum_{j=1}^{k-1}\mathbb E_{n,\rho}[\widetilde Z_j \mid \mathcal F_{j-1}]$ is $\mathcal F_{k-1}$-measurable, and 
\[
\mathbb E_{n,\rho}\left[\mathbb E_{n,\rho}[\widetilde Z_k \mid \mathcal F_{k-1}] \middle| \mathcal F_{k-1}\right] = \mathbb E_{n,\rho}[\widetilde Z_k \mid \mathcal F_{k-1}],
\]
so the $\widetilde Z_k$ terms cancel. Since $|\widetilde Z_j|\leq C$ for all $j\leq k$, each $\widetilde Z_j$ is $L^1$-integrable. Therefore, $Y_k^{(\rho)}\in L^1(\Pbb_{n,\rho})$ since it is a finite sum of integrable random variables. Thus $\{Y_k^{(\rho)}\}_{k \geq 0}$ is a martingale under $\Pbb_{n,\rho}$ with respect to $\{\mathcal{F}_k\}_{k\geq 0}$.
\end{IEEEproof}
% Altogether, we conclude that $\{L_k(\sigma)\}_{k \geq 0}$, $\{\widetilde L_k\}_{k\ge0}$ and $\{Y_k^{(\rho)}\}_{k \geq 1}$ are martingales under $\Pbb_{n,\rho}$ with respect to the filtration $\{\mathcal{F}_k\}_{k\geq 0}$, and $\{Y_k^{(\sigma)}\}_{k\ge1}$ is a martingale under $\Pbb_{n,\sigma}$ with respect to the same filtration.
\subsection{Proof of \Cref{lemma:mixture_martingale_comparison_with_pointwise}}\label{appendix:proof_of_lemma_9}
Before proving \Cref{lemma:mixture_martingale_comparison_with_pointwise}, we establish two bounds used in the argument. We start with bounding the absolute difference between two pointwise SQLLRs.
\begin{lemma}\label{lemma:covering}
\looseness=-1
Let $\delta$ be the smallest eigenvalue attained by any state in $\mathcal D$, which is strictly positive since $\mathcal D$ is compact and consists of full-rank states. Then, for fixed $\sigma,\sigma'\in\mathcal D$ and $k\geq 1$, we have
\begin{equation}
|S_{k}(\sigma)-S_{k}(\sigma')|\leq k\log\left(1+\frac{\|\sigma-\sigma'\|_1}{\delta}\right).
\end{equation}
\end{lemma}
\begin{IEEEproof}
Fix $\sigma,\sigma'\in\mathcal D$ and $k\geq1$. Then we have
\begin{align*}
\hspace{-0.22cm }|S_{k}(\sigma)-S_{k}(\sigma')|
&=\left|\sum_{j=1}^{k}\log\frac{\tr\big[M_j(X_j)\sigma'\big]}{\tr\big[M_j(X_j)\sigma\big]}\right| \nonumber \\
% &=\biggl|\sum_{j=1}^{k}Z_j(\sigma)-Z_j(\sigma')\biggr|  \label{eq:triangle_inequality} \\
&\leq \sum_{j=1}^{k}\biggl|\log\frac{\tr\big[M_j(X_j)\sigma'\big]}{\tr\big[M_j(X_j)\sigma\big]}\biggr| \\
&\leq k\max\{D_{\max}\left(\sigma \| \sigma'\right),D_{\max}\left(\sigma' \| \sigma\right)\} \\
&\leq k\log\left(1+\frac{\|\sigma-\sigma'\|_1}{\delta}\right)
\end{align*}
where $D_{\max}(\sigma\|\sigma') := \log \inf\left\{ \lambda>0 : \sigma \le \lambda \sigma' \right\}$. For the first inequality we used the definition of $S_k(\cdot)$ and the triangle inequality, for the second we used \cite[Lemma 15]{li2022optimal} and for the last we used that $\|\sigma-\sigma'\|_\infty\leq\|\sigma-\sigma'\|_1$, thus $\sigma-\sigma'\preceq \|\sigma-\sigma'\|_1I$, which implies that
\begin{equation*}
\sigma\preceq\sigma'+\|\sigma-\sigma'\|_1I\preceq\sigma'+\frac{\|\sigma-\sigma'\|_1}{\delta}\sigma'=\biggl(1+\frac{\|\sigma-\sigma'\|_1}{\delta}\biggr)\sigma'
\end{equation*}
since every state in $\mathcal D$ satisfies the constraint $\sigma\succeq \delta I$. Similarly,  $\sigma'\preceq\bigl(1+\tfrac{\|\sigma-\sigma'\|_1}{\delta}\bigr)\sigma.$ This concludes the proof.
\end{IEEEproof}

The above result gives a bound, linear in the time step $k$, on the absolute difference between two pointwise-SQLLRs. The next lemma gives the second bound needed for the proof of \Cref{lemma:mixture_martingale_comparison_with_pointwise}: a polynomial lower bound on the prior mass of small trace-norm balls in $\mathcal D$.
\begin{figure}[!t]
    \centering
    \begin{tikzpicture}
        \draw(-2,-2)rectangle(2,2);
        \draw(3.5,-2)rectangle(6.5,2);
        \node at(0,1.8){$\mathsf A_d$};
        \node at(5,1.8){$\mathbb R^{d^2-1}$};
        \draw[magenta](0,0) circle [x radius=2cm, y radius=1cm, rotate=30];
        \draw[magenta](5,0) circle[x radius=1cm, y radius=14mm, rotate=60];\fill[magenta,opacity=.15](0,0) circle [x radius=2cm, y radius=1cm, rotate=30];
        \fill[magenta,opacity=.15](5,0) circle[x radius=1cm, y radius=14mm, rotate=60];
        \draw[-stealth](2.2,.1)--(3.3,.1)node[midway,above]{$x$};
        \draw[-stealth](3.3,-.1)--(2.2,-.1)node[midway,below]{$T$};
        \fill[red](.5,.7)circle[radius=.05];
        \draw[darkgreen](.5,.7)circle[radius=.5];
        \fill[red](5.5,.5)circle[radius=.05];
        \draw[blue](5.5,.5)circle[radius=.2];
        \draw[blue](.5,.7)circle[x radius=4mm,y radius=2mm,rotate=100];
        \node[magenta]at(1.5,1.5){$\mathcal D$};
        \node[magenta]at(5.7,1.15){$x(\mathcal D)$};
        \draw(.5,.7)--(-1,1)node[red,above]{$\sigma^{(k)}$};
        \draw(.552, .305)--(.552,-1.2)node[blue,below]{$T(\mathcal B_{\Euc}(x(\sigma^{(k)}),a\varepsilon))$};
        \draw(.066,.452)--(-.9,-.1)node[darkgreen,below]{$\mathcal B_1(\sigma^{(k)},\varepsilon)$};
        \draw(5.5,.5)--(5,.5)node[red,left]{$x(\sigma^{(k)})$};
        \draw(5.444,0.308)--(5,-1.2)node[blue,below]{$\mathcal B_{\Euc}(x(\sigma^{(k)}),a\varepsilon)$};
    \end{tikzpicture}
    \caption{Illustration of the second part in the proof of \Cref{lemma:balls}. The scalar $a$ comes from the norm-equivalence bound in \Cref{lemma:isomorphism}, while for a fixed $\sigma^{(k)}\in\interior(\mathcal D)$, $\varepsilon$ is chosen so that \eqref{eq:lower_bound_balls} is satisfied.}
    \label{fig:balls}
\end{figure}
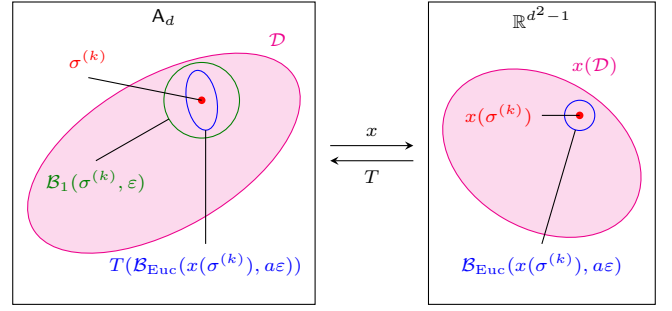
\begin{lemma}\label{lemma:balls} 
Fix once and for all an auxiliary reference state
$\sigma^\circ\in\interior_{\mathsf A_d}(\mathcal D)$. Then for every $k\geq1$ and any fixed $\sigma\in \mathcal D$, define the sequence $\{\sigma^{(k)}\}_{k\ge1}$ as
\[
\sigma^{(k)}:=\bigg(1-\frac{1}{k}\bigg)\sigma+\frac{1}{k}\sigma^\circ.
\]
There exist a constant $v_0>0$, depending only on the fixed choice of
$\sigma^\circ$, and a positive sequence $(\varepsilon_k)_{k\ge1}$, independent of $\sigma$, such that for all $k\ge1$
and all $0<\varepsilon\le \varepsilon_k$,
\begin{equation}
 \Pi(\mathcal B_1(\sigma^{(k)},\varepsilon))
 \geq v_0\varepsilon^{d^2-1}.
\end{equation}
\end{lemma}

For clarity, we briefly outline the structure of the argument in three parts. First, we introduce a coordinate representation of $\mathsf A_d$, denoted by $x$, and use \Cref{lemma:isomorphism} to control the trace-norm on $\mathsf A_d$ with the Euclidean norm on $\mathbb R^{d^2-1}$. Next, we construct a trace-norm neighborhood in $\mathcal D$, $\mathcal B_1(\sigma^{(k)},\varepsilon)$, and a Euclidean neighborhood of $x(\sigma^{(k)})$, which lies inside $x(\mathcal D)$. Then we map the Euclidean balls back to $\mathcal D$, which have been constructed such that they are subsets of the trace-norm neighborhoods $\mathcal B_1(\sigma^{(k)},\varepsilon)$. \Cref{fig:balls} gives a visual aid of  this part. Finally, combining both these steps enables us to lower bound $\Pi(\mathcal B_1(\sigma^{(k)},\varepsilon))$ by the volume of the Euclidean neighborhoods. Using the standard formula for the volume of a $(d^2-1)$-dimensional Euclidean ball together with some basic properties of Lebesgue measure gives us the desired result.

\begin{IEEEproof}
We now begin with the first step: introducing a coordinate representation of $\mathsf A_d$, which we subsequently restrict to $\mathcal D$, and establishing a comparison between the trace-norm and the Euclidean norm using \Cref{lemma:isomorphism}.

Note that the associated translation space $\mathsf V_d$ is a real vector space of real dimension $d^2-1$, thus $\mathsf V_d \cong \mathbb{R}^{d^2-1}$. Therefore, there exists an affine isomorphism $T : \mathbb{R}^{d^2-1} \to \mathsf A_d$ such that each $\sigma \in \mathsf A_d$ admits a coordinate representation
\[x(\sigma):=T^{-1}(\sigma) \in \mathbb{R}^{d^2-1}.\]
As $x$ is an affine isomorphism, its linear part $J:\mathsf V_d\to \mathbb R^{d^2-1}$ is a linear isomorphism. Equivalently, the affine map $T$ has linear part $J^{-1}$.  Thus for all $\sigma,\sigma'\in\mathsf A_d$,
\[
x(\sigma)-x(\sigma')=J(\sigma-\sigma').
\]
Applying \Cref{lemma:isomorphism} yields that there exists $a>0$ such that, for all $\sigma,\sigma'\in\mathsf A_d$,
\[
a\|\sigma-\sigma'\|_1 \le \|J(\sigma-\sigma')\|_2= \|x(\sigma)-x(\sigma')\|_2.
\]
In particular, this holds for all $\sigma,\sigma'\in\mathcal D$ as $\mathcal D$ is a subset of $\mathsf A_d$. This concludes the first part.

We now construct a Euclidean neighborhood around $x(\sigma^{(k)})$ whose image under $T$ lies inside the trace-norm ball $\mathcal B_1(\sigma^{(k)},\varepsilon)$. This enables us to embed a Euclidean ball of known volume inside $\mathcal B_1(\sigma^{(k)},\varepsilon)$, yielding a lower bound on $\Pi(\mathcal B_1(\sigma^{(k)},\varepsilon))$ . See \Cref{fig:balls} for an illustration.

Since $x$ is an affine isomorphism between finite-dimensional normed spaces, it is automatically continuous with a continuous inverse, and therefore a homeomorphism. In particular, $x$ preserves interior points and consequently $x(\sigma^\circ) \in \interior_{\mathbb R^{d^2-1}}(x(\mathcal D))$. Therefore, there exists a small radius $r_0>0$ such that $\mathcal B_{\Euc}(x(\sigma^\circ), r_0) \subseteq x(\mathcal D)$. Now fix any $\sigma\in\mathcal D$ and define the sequence $\{\sigma^{(k)}\}_{k\ge1}$ as
\[
\sigma^{(k)}:=\bigg(1-\frac{1}{k}\bigg)\sigma+\frac{1}{k}\sigma^\circ,
\]
where $k\geq1$. By definition $\sigma^{(k)}\in \interior_{\mathsf A_d}(\mathcal D)$, since 
$\sigma \in \mathcal D$ and $\sigma^{(k)}$ is a convex combination with strictly positive weight on $\sigma^\circ\in\interior_{\mathsf A_d}(\mathcal D)$. Since $x(\sigma^{(k)}) = \left(1-\frac1k\right)x(\sigma) + \frac1k x(\sigma^\circ)$ and $x(\mathcal D)$ is convex as $\mathcal D$ is convex and $x$ is affine, then it follows that
\begin{align*}
\mathcal B_{\Euc}\left(x(\sigma^{(k)}), \tfrac{r_0}{k}\right)
&= x(\sigma^{(k)}) + \mathcal B_{\Euc}\!\left(0, \tfrac{r_0}{k}\right) \\
&= \left(1-\tfrac1k\right)x(\sigma) + \tfrac1k x(\sigma^\circ) + \tfrac1k\mathcal B_{\Euc}\left(0, r_0\right) \\
&= \left(1-\tfrac1k\right)x(\sigma) + \tfrac1k\big(x(\sigma^\circ) + \mathcal B_{\Euc}(0, r_0)\bigr) \\
&= \left(1-\tfrac1k\right)x(\sigma) + \tfrac1k \mathcal B_{\Euc}(x(\sigma^\circ), r_0) \\
&\subseteq \left(1-\tfrac1k\right)x(\sigma) + \tfrac1k x(\mathcal D) \\
&\subseteq x(\mathcal D).
\end{align*}
where we used translation invariance and scaling properties of Euclidean balls, the inclusion $\mathcal B_{\Euc}(x(\sigma^\circ), r_0) \subseteq x(\mathcal D)$, and the convexity of $x(\mathcal D)$.
In particular, for all $0<\varepsilon \le \varepsilon_k$, where $\{\varepsilon_k\}_{k\ge1}$ is a sequence \emph{of our choice}  satisfying $\varepsilon_k \le \frac{r_0}{ak}$ for all $k$, applying the affine map $T$ gives
\[
T(\mathcal B_{\Euc}(x(\sigma^{(k)}), a\varepsilon)) \subset \mathcal D,
\]
and therefore $\Pi(T(\mathcal B_{\Euc}(x(\sigma^{(k)}), a\varepsilon)))$ is well defined. Now let $\sigma' \in T\big(\mathcal B_{\Euc}(x(\sigma^{(k)}), a\varepsilon)\big)$. Then $\|x(\sigma') - x(\sigma^{(k)})\|_2 \le a\varepsilon$ and since $x(\sigma') - x(\sigma^{(k)}) = J(\sigma' - \sigma^{(k)})$, \Cref{lemma:isomorphism} gives
\[
a \|\sigma' - \sigma^{(k)}\|_1 \le \|x(\sigma') - x(\sigma^{(k)})\|_2
\le a\varepsilon,
\]
which yields that $\|\sigma' - \sigma^{(k)}\|_1 \le \varepsilon$. Thus we have
\begin{equation}
T(\mathcal B_{\Euc}(x(\sigma^{(k)}), a\varepsilon))\subseteq \mathcal B_1(\sigma^{(k)},\varepsilon) \subseteq \mathcal D,\label{eq:lower_bound_balls}
\end{equation}
where the second inclusion  is immediate from the definition of $\mathcal B_1(\sigma^{(k)},\varepsilon)$ as a trace-norm ball in $\mathcal D$. This completes the second part. 

To conclude the proof, we combine the above inclusions to obtain the desired bound. 

Using \eqref{eq:lower_bound_balls} and the monotonicity of $\Pi$ we have
\begin{align}
\hspace{-0.4cm}\Pi(\mathcal B_1(\sigma^{(k)},\varepsilon))&\geq\Pi(T(\mathcal B_{\Euc}(x(\sigma^{(k)}), a\varepsilon))) \nonumber \\
&=\frac{\Vol_{\mathsf{A}_d}(T(\mathcal B_{\Euc}(x(\sigma^{(k)}), a\varepsilon)))}{\Vol_{\mathsf{A}_d}(\mathcal D)} \nonumber \\
&=\frac{|\det J^{-1}|\cdot\Vol_{\mathbb R^{d^2-1}}(\mathcal B_{\Euc}(x(\sigma^{(k)}), a\varepsilon))}{\Vol_{\mathsf{A}_d}(\mathcal D)} \nonumber \\ 
& = \frac{|\det J^{-1}|\Vol_{\mathbb R^{d^2-1}}(\mathcal B_{\Euc}(0,1))}{\Vol_{\mathsf{A}_d}(\mathcal D)}(a\varepsilon)^{d^2-1} \nonumber \\
% & = \frac{|\det A^{-1}|\cdot\pi^{\frac{d^2-1}{2}}}{\Gamma(\frac{d^2-1}{2}+1)\Vol_{\state}(\mathcal D)}(c\varepsilon)^{d^2-1} \nonumber \\
&=v_{0}\varepsilon^{d^2-1} \nonumber 
\end{align}
where in the first equality we use the definition of $\Pi$ and in the second equality we used the change-of-variables formula: for any Borel measurable $S\subset \mathbb R^{d^2-1}$ we have,
\[
\Vol_{\mathsf{A}_d}(T(S)) = |\det J^{-1}|\Vol_{\mathbb R^{d^2-1}}(S),
\]
where $|\det J^{-1}|$ is the Jacobian determinant of $T$. Then, in the third equality we used translation invariance and homogeneity of $\Vol_{\mathbb{R}^{d^2-1}}$, and in the last equality we used the standard formula for the volume of the Euclidean unit ball \cite{SmithVamanamurthy1989} and define $$v_{0}:=\frac{|\det J^{-1}|\cdot\pi^{\frac{d^2-1}{2}}a^{d^2-1}}{\Gamma(\frac{d^2-1}{2}+1)\Vol_{\mathsf{A}_d}(\mathcal D)},$$ 
which is a finite positive constant independent of $k$ and $\sigma$. This concludes the proof.
\end{IEEEproof}
With these two bounds acquired, we now prove \Cref{lemma:mixture_martingale_comparison_with_pointwise}. 
\begin{IEEEproof}
For $k\geq1$, define the sequence $\{\sigma^{(k)}\}_{k\ge1}$ as in \Cref{lemma:balls}. Namely, for every $k\geq1$,
\[
\sigma^{(k)}:=\bigg(1-\frac{1}{k}\bigg)\sigma+\frac{1}{k}\sigma^\circ.
\]
Then, for all $\sigma' \in \mathcal B_1(\sigma^{(k)},\varepsilon)$ with $0<\varepsilon \le \varepsilon_k$, where $\{\varepsilon_k\}_{k\ge1}$ is a sequence \emph{of our choice} satisfying $\varepsilon_k \le \frac{r_0}{ak}$ for all $k$, we have the following chain of relationships
\begin{align*}
\widetilde L_k&= \int_{\mathcal D}e^{-S_k(\sigma')}\Pi(\dd\sigma')\\ &\geq \int_{\mathcal B_1(\sigma^{(k)},\varepsilon)}e^{-S_k(\sigma')}\Pi(\dd\sigma')\\
&\geq \int_{\mathcal B_1(\sigma^{(k)},\varepsilon)} e^{-S_k(\sigma^{(k)})-k\log\left(1+\frac{\varepsilon}{\delta}\right)} \Pi(\dd\sigma') \\
&= e^{-S_k(\sigma^{(k)})-k\log\left(1+\frac{\varepsilon}{\delta}\right)} \Pi(\mathcal B_1(\sigma^{(k)},\varepsilon)) \\
&\geq v_{0}\varepsilon^{d^2-1} e^{-S_k(\sigma^{(k)})-k\log\left(1+\frac{\varepsilon}{\delta}\right)} \label{eq:partv}
\end{align*}
where in the first inequality we used that $\mathcal B_1(\sigma^{(k)},\varepsilon)\subseteq \mathcal D$ and $e^{-S_k(\cdot)}\geq0$, so restricting the domain can only decrease the integral. The second inequality follows from \Cref{lemma:covering} with $\|\sigma'-\sigma^{(k)}\|_1\leq \varepsilon$ and the last inequality uses \Cref{lemma:balls}. Applying $-\log$ to this lower bound on $\widetilde L_k$ and rearranging gives
\begin{align*}
\widetilde S_k-S_k(\sigma^{(k)})\leq k\log\left(1+\frac{\varepsilon} {\delta}\right) -\log v_{0} - (d^2-1)\log \varepsilon.
\end{align*}
Choosing $\varepsilon=\varepsilon_{k}:=\min\left\{\frac{r_0}{ak},\frac{\delta}{k}\right\}$ gives 
\begin{align*}
\hspace{-0.22cm} \widetilde S_k-S_k(\sigma^{(k)}) &\leq k\log\left(1+\frac{\varepsilon_k} {\delta}\right) -\log v_{0} - (d^2-1)\log \varepsilon_k\\
&\leq 1 -\log v_{0} - (d^2-1)\log \varepsilon_k \\
% &=  1 -\log v_{0} -(d^2-1)\biggl(\log\bigl(\min\{\delta,\tfrac{r_0}{a}\}\bigr) - \log k\biggr)\\
&=C_1+(d^2-1)\log k
\end{align*}
where $C_1:=1 -\log v_{0}-(d^2-1)\log\bigl(\min\{\delta,\tfrac{r_0}{a}\}\bigr)<\infty$ and we used that since $\varepsilon_{k}\leq\frac{\delta}{k}$, then
\[
k\log\biggl(1+\frac{\varepsilon_k}{\delta}\biggr)\leq k\log\biggl(1+\frac{\delta/k}{\delta}\biggr)= k\log\biggl(1+\frac{1}{k}\biggr)
\leq 1.
\]
% $k\log(1+\frac{\varepsilon_k} {\delta})\leq k\log(1+\frac{1}{k})
% \leq 1$, since  $\varepsilon_{k}\leq\frac{\delta}{k}$. Suppose that $\varepsilon_k=\frac{r_0}{ak}$, then
% \begin{align*}
% \widetilde S_k-S_k(\sigma^{(k)}) &\leq 1 -\log v_{0} - (d^2-1)\log \varepsilon_k\\
% &=C_1+(d^2-1)\log k
% \end{align*}
% where $C_1:=1-\log v_{0}-(d^2-1)\log\frac{r_0}{a}.$ If $\varepsilon_k=\frac{\delta}{k}$ then 
% \begin{align*}
% \widetilde S_k-S_k(\sigma^{(k)}) &\leq 1 -\log v_{0} - (d^2-1)\log \varepsilon_k\\
% &=C_2 +(d^2-1)\log k
% \end{align*}
% where $ C_2:=1 -\log v_{0}-(d^2-1)\log \delta<\infty$. Altogether we get for $\sigma \in  \mathcal D$,
% \begin{align*}
% \hspace{-0.22cm}\widetilde S_k-S_k(\sigma^{(k)}) 
% &\leq\max\{C_1, C_2\} +(d^2-1)\log k \\
% &\leq C_3+(d^2-1)\log k
% \end{align*}
% where $C_3:=\max\{C_1, C_2 \}<\infty$.
Finally, to prove the statement for all $\sigma \in \mathcal D$, note that $\|\sigma^{(k)} - \sigma\|_1=\tfrac 1k\|\sigma^\circ - \sigma\|_1\leq \tfrac 2k$, then by \Cref{lemma:covering} we have
% \begin{align*}
% |S_k(\sigma^{(k)}) - S_k(\sigma)|&\leq k\log\biggl(1+\frac{2}{k\delta}\biggr)
% \leq\frac{2}{\delta}
% \end{align*}
% $
% |S_k(\sigma^{(k)}) - S_k(\sigma)|\leq\frac{2}{\delta}$
% \begin{align*}
% |S_k(\sigma^{(k)}) - S_k(\sigma)| &\le k \log\left(1 + \frac{\|\sigma^{(k)} - \sigma\|_1}{\delta}\right)\\
% &\le k \log\left(1 + \frac{2}{k\delta}\right)\leq \frac{2}{\delta}, 
% \end{align*}
% and thus 
$S_k(\sigma^{(k)}) \leq S_k(\sigma)+\frac{2}{\delta}$. Therefore, setting $\widetilde C:=C_1+\frac{2}{\delta}$ we conclude
\begin{align*}
\widetilde S_k - S_k(\sigma) &= \bigl(\widetilde S_k - S_k(\sigma^{(k)})\bigr) + \bigl(S_k(\sigma^{(k)}) - S_k(\sigma)\bigr) \\ &\le C_1 + (d^2-1)\log k + \frac{2}{\delta}\\
&=\widetilde C + (d^2-1)\log k.
\end{align*}

This completes the argument. 
\end{IEEEproof}
\subsection{Proof of \Cref{lemma:lemma19}}\label{appendix:proof_of_lemma_10}
In this section, we provide the detailed proof of \Cref{lemma:lemma19}, proceeding by establishing each part separately.
\subsubsection{Proof of \Cref{lemma:lemma19} \ref{lemma:19;part1}}
Since $\widetilde Z_k$ is a function of $X_k$, and  $M_k$ is $\mathcal F_{k-1}$-measurable, conditioning on $\mathcal F_{k-1}$ and using that $\mathcal X$ is finite yields,
\begin{align}
\hspace{-0.22cm} \mathbb E_{n,\rho}[ \widetilde Z_k \mid \mathcal F_{k-1}]
% &= \sum_{x\in\mathcal X}\Pbb_{n,\rho}(X_k=x|\mathcal F_{k-1})\log\frac{\tr[M_k(x)\rho]}{\tr[M_k(x)\tilde\sigma_{k-1}]} \nonumber \\
% &= \sum_{x\in\mathcal X}\tr[M_k(x)\rho]\log\frac{\tr[M_k(x)\rho]}{\tr[M_k(x)\tilde\sigma_{k-1}]} \nonumber \\
&= D\left(P_{\rho,M_k}\|P_{\tilde\sigma_{k-1},M_k}\right) \nonumber \\
&= \mathbf 1_{\{\widetilde S_{k-1}\ge 0\}}D_\mathcal{M}(\rho\|\tilde\sigma_{k-1}) \nonumber  \\
&+ \mathbf 1_{\{\widetilde S_{k-1}< 0\}}D\bigl(P_{\rho,m_{1}^\ast(\tilde\sigma_{k-1})}\|P_{\tilde\sigma_{k-1},m_{1}^\ast(\tilde\sigma_{k-1})} \bigr). \nonumber
\end{align}
Now, since $\rho\notin\mathcal D$, $\tilde\sigma_{k-1}\in \mathcal D$, $\mathcal D$ is compact, $\sigma\mapsto D_\mathcal{M}(\rho\|\sigma)$ is continuous on the
compact set $\mathcal D$  and $D_\mathcal{M}(\rho\|\sigma)>0$ for every $\sigma\in\mathcal D$ because $\mathcal{M}_\mathcal{X}$ separates quantum states, then $D_{\mathcal M}(\rho\|\mathcal D)=\eta_1>0$ with the constant $\eta_1$ being independent of $n$ and $k$. Thus $\mathbb E_{n,\rho}[ \widetilde Z_k \mid \mathcal F_{k-1}]$ is uniformly bounded below when $\widetilde S_{k-1}\ge 0$. If instead $\widetilde S_{k-1}< 0$, then
\begin{align*}
\mathbb E_{n,\rho}[ \widetilde Z_k \mid \mathcal F_{k-1}]
&= D\left(P_{\rho,m_{1}^\ast(\tilde\sigma_{k-1})}\|P_{\tilde\sigma_{k-1},m_{1}^\ast(\tilde\sigma_{k-1})} \right)\\
&\ge \inf_{\sigma \in \mathcal D} D\left(P_{\rho,m_{1}^\ast(\sigma)}\|P_{\sigma,m_{1}^\ast(\sigma)} \right)
=: \eta_2.  
\end{align*}
where $\eta_2$ is a non-negative constant independent of $n$ and $k$. Suppose that $\eta_2=0$, then there exists a sequence $\{\sigma_i\}_{i\ge1}\subset\mathcal D$ such that
\[
D\left(P_{\rho,m_1^\ast(\sigma_i)}\middle\|P_{\sigma_i,m_1^\ast(\sigma_i)}\right)\to 0.
\] 
By compactness of $\mathcal D$, there exists a subsequence $\{\sigma_{i_j}\}_{j\ge1}$ such that $\sigma_{i_j}\to \sigma\in\mathcal D$. Furthermore, by the compactness of $\mathcal M_{\mathcal X}$, there exists a subsequence of measurements $m_1^\ast(\sigma_{i_{j_\ell}})\in \mathcal M_{\mathcal X}$ such that $m_1^\ast(\sigma_{i_{j_\ell}})\to m^  \star \in  \mathcal M_{\mathcal X}$. 
Since for every $m\in\mathcal M_{\mathcal X}$ and all $\ell$, we have by the definition of $m_1^\ast(\cdot)$
\[
D(P_{\sigma_{i_{j_\ell}},m_1^\ast(\sigma_{i_{j_\ell}})}\|P_{\rho,m_1^\ast(\sigma_{i_{j_\ell}})}) \ge D(P_{\sigma_{i_{j_\ell}},m}\|P_{\rho,m}).
\]
Passing to the limit as $\ell\to\infty$ and using continuity of 
$(\sigma,m)\mapsto D(P_{\sigma,m}\|P_{\rho,m})$, which holds since $\rho$ and every $\sigma\in\mathcal D$ are full-rank, so after discarding zero-outcome effects the induced probabilities are strictly positive, then yields $D(P_{\sigma,m^\ast}\|P_{\rho,m^\ast})\ge D(P_{\sigma,m}\|P_{\rho,m})$ for all $m\in\mathcal M_{\mathcal X}$. Thus $m^\ast \in \arg\max_m D(P_{\sigma,m}\|P_{\rho,m})$. Moreover, $(\sigma,m)\mapsto D(P_{\rho,m}\|P_{\sigma,m})$ is continuous in $(\sigma,m)$ so we have
\[
D\left(P_{\rho,m^\ast}\middle\|P_{\sigma,m^\ast}\right)=\lim_{\ell\to \infty}D\left(P_{\rho,m_1^\ast(\sigma_{i_{j_\ell}})}\middle\|P_{\sigma_{i_{j_\ell}},m_1^\ast(\sigma_{i_{j_\ell}})}\right)=0,
\]
which implies $P_{\rho,m^\ast}=P_{\sigma,m^\ast}$. Therefore $D\left(P_{\sigma,m^\ast}\middle\|P_{\rho,m^\ast}\right)=0$, and
since $m^\ast$ is a maximizer it follows that for all $m$,
\[
D(P_{\sigma,m}\|P_{\rho,m})\leq D\left(P_{\sigma,m^\ast}\middle\|P_{\rho,m^\ast}\right)=0.
\]
Thus, $D(P_{\sigma,m}\|P_{\rho,m})=0$ and hence $P_{\sigma,m}=P_{\rho,m}$ for all $m$. Since the class $\mathcal{M}_\mathcal{X}$ separates quantum states, it follows that $\sigma=\rho$. This contradicts $\rho\notin\mathcal D$. Hence $\eta_2>0$. Therefore taking $\eta:=\min\{\eta_1,\eta_2\}>0$ concludes the argument.

\subsubsection{Proof of \Cref{lemma:lemma19} \ref{lemma:19;part2}} We first write
 \[\mathbb E_{n,\rho}\left[e^{-\lambda\widetilde S_k}\right]=\mathbb E_{n,\rho}\left[e^{-\lambda\widetilde S_{k-1}}\mathbb E_{n,\rho}\left[e^{-\lambda\widetilde Z_k}\mid\mathcal F_{k-1}\right]\right],\]
which follows since $\widetilde S_{k}=\widetilde S_{k-1}+\widetilde Z_k$ and $\widetilde S_{k-1}$ is $\mathcal F_{k-1}$-measurable. We next evaluate the conditional expectation.
\begin{align*}
\hspace{-0.22cm} \mathbb E_{n,\rho}\left[e^{-\lambda\widetilde Z_k}|\mathcal F_{k-1}\right]
% &= \sum_{x\in\mathcal X}\Pbb_{n,\rho}(X_k=x|\mathcal F_{k-1})
% \left(\frac{P_{\tilde\sigma_{k-1},M_k}(x)}{P_{\rho,M_k}(x)}\right)^\lambda \\
&= \sum_{x\in\mathcal X}
P_{\rho,M_k}(x)\biggl(\frac{P_{\tilde\sigma_{k-1},M_k}(x)}{P_{\rho,M_k}(x)}\biggr)^\lambda
 \\
&= \sum_{x\in\mathcal X}
P_{\rho,M_k}(x)^{1-\lambda}
P_{\tilde\sigma_{k-1},M_k}(x)^\lambda.
\end{align*}
Applying Young's inequality with $p=1/(1-\lambda)$ and $q=1/\lambda $ to each term gives
\[ 
\mathbb E_{n,\rho}\left[e^{-\lambda\widetilde Z_k}\mid\mathcal F_{k-1}\right]\le 1 \] 
with equality if and only if $P_{\rho,M_k}=P_{\tilde\sigma_{k-1},M_k}$. This equality case is uniformly excluded, otherwise the conditional drift of $\widetilde Z_k$ under $\mathbb P_{n,\rho}$ would be zero, contradicting \Cref{lemma:lemma19} \ref{lemma:19;part1}. Hence, by the same compactness argument used there, the conditional overlap is uniformly bounded away from one, so there exists a constant $0<c<1$, independent of $k$ and $n$, such that
\[
\mathbb E_{n,\rho}\left[e^{-\lambda\widetilde Z_k}\mid\mathcal F_{k-1} \right]\le c\quad\text{a.s.} 
\] 
Therefore, 
\begin{align*} \mathbb E_{n,\rho}\left[e^{-\lambda \widetilde S_k}\right]&= \mathbb E_{n,\rho}\left[ e^{-\lambda \widetilde S_{k-1}} \mathbb E_{n,\rho}\left[e^{-\lambda \widetilde Z_k} \mid \mathcal F_{k-1}\right] \right] \\&\le c \mathbb E_{n,\rho}\left[e^{-\lambda \widetilde S_{k-1}}\right]. 
\end{align*} 
Iterating this $k$ times with $e^{-\lambda\widetilde S_0}=1$ gives the first result. The second result then follows by Markov’s inequality.

\subsubsection{Proof of \Cref{lemma:lemma19} \ref{lemma19;part3}}
For every $n$, define for each $k\geq0$ a probability measure $\mathbb Q_n^{(k)}$ on $(\Omega,\mathcal F_k)$ by specifying, for each \(k\ge0\), its
Radon--Nikodym derivative with respect to \(\Pbb_{n,\rho}\) on
\(\mathcal F_k\):
\[
\left.\frac{\dd\mathbb Q_n^{(k)}}{\dd\Pbb_{n,\rho}}\right|_{\mathcal F_k}
:= \widetilde L_k, \qquad k\geq0 .
\]
This family of measures is consistent, since \(\{\widetilde L_k\}_{k\geq0}\) is a non-negative
\(\Pbb_{n,\rho}\)-martingale with
\(\mathbb E_{n,\rho}[\widetilde L_k]=1\). Indeed, for \(A\in\mathcal F_{k-1}\),
\begin{align*}
\mathbb Q_n^{(k)}(A)&=\mathbb E_{n,\rho}\left[\widetilde L_k\mathbf 1_A\right]\\
&=\mathbb E_{n,\rho}\left[ \mathbb E_{n,\rho}[\widetilde L_k\mid\mathcal F_{k-1}]\mathbf 1_A \right] \\
&=
\mathbb E_{n,\rho}\left[\widetilde L_{k-1}\mathbf 1_A\right]=\mathbb Q_n^{(k-1)}(A).
\end{align*}
Since $\Omega=(\mathcal M_{\mathcal X}\times\mathcal X)^{\infty}$ is equipped with the usual product topology and $\mathcal F$ is the Borel $\sigma$-algebra generated by this topology, with $\mathcal F=\sigma(\bigcup_{k\ge0}\mathcal F_k)$, this consistent family uniquely extends to a probability measure $\mathbb Q_n$ on $(\Omega,\mathcal F)$ such that
\begin{equation}
\left.\frac{\dd\mathbb Q_n}{\dd\Pbb_{n,\rho}}\right|_{\mathcal F_k}
= \widetilde L_k, \quad \forall k\geq0 .
\end{equation}
Equivalently, $\mathbb Q_n$ admits the predictive representation 
\[
\mathbb{Q}_n\left( X_k = x \mid \mathcal{F}_{k-1}, M_{k} \right)
= \tr\left[ M_k(x)\tilde\sigma_{k-1} \right],
\]
for all $k \ge 1$ and $x \in \mathcal X$.

We now state an analogue of \Cref{lemma:lemma19} \ref{lemma:19;part1} and \ref{lemma:19;part2}.
\begin{lemma} \label{lemma:alternative_lemma}
Under the probability $\mathbb Q_n$, the following properties hold:
\begin{enumerate}[label=(\roman*)]
\item\label{lemma:alternative_lemma;part1}
The conditional expectation of \(-\widetilde Z_k\) given \(\mathcal F_{k-1}\) satisfies
\begin{multline}
\mathbb E_{\mathbb Q_n}\left[-\widetilde Z_k \mid \mathcal F_{k-1}\right]
=\mathbf 1_{\{\widetilde S_{k-1}< 0\}} D_{\mathcal M}(\tilde\sigma_{k-1}\|\rho) \\
+\mathbf 1_{\{\widetilde S_{k-1}\ge 0\}}
D\left(P_{\tilde\sigma_{k-1},m_{0}^\ast(\tilde\sigma_{k-1})}\middle\|P_{\rho,m_{0}^\ast(\tilde\sigma_{k-1})}\right).
\end{multline}
In particular, 
% there exists a constant \(\eta>0\), independent of \(k\) and \(n\), such that
% \[
% \mathbb E_{\mathbb Q_n}\!\left[-\widetilde Z_k \mid \mathcal F_{k-1}\right]\ge \eta
% \qquad \text{a.s.}
% \]
% and hence 
the process $\{-\widetilde S_k \}_{k\ge1}$ has uniformly positive conditional drift under $\mathbb Q_n$.
\item\label{lemma:alternative_lemma;part2}
For sufficiently small \(\lambda\), there exists a constant \(0<c<1\) such that
\begin{align}
\mathbb E_{\mathbb Q_n}\left[e^{\lambda \widetilde S_k}\right]\le c^k
\text{ and }
\mathbb Q_n(\widetilde S_k\ge 0)\le c^k.
\end{align}
\end{enumerate}
\end{lemma}
The proof follows by a direct analogue of the proof of \Cref{lemma:lemma19} \ref{lemma:19;part1} and \ref{lemma:19;part2}, upon replacing $\Pbb_{n,\rho}$ with $\mathbb Q_n$ and interchanging the roles of $\rho$ and $\tilde\sigma_{k-1}$. The same martingale and conditional expectation arguments apply verbatim, and we therefore omit the details. 

We now prove \Cref{lemma:lemma19} \ref{lemma19;part3}. For fixed \(\sigma\in\mathcal D\), we have using the chain rule
\[
\left.\frac{\dd\Pbb_{n,\sigma}}{\dd\mathbb Q_n}\right|_{\mathcal F_k}
= \frac{\dd\Pbb_{n,\sigma}}{\dd\Pbb_{n,\rho}} \frac{\dd\Pbb_{n,\rho}}{\dd\mathbb Q_n} 
=
\frac{L_k(\sigma)}{\widetilde L_k}
=
e^{\widetilde S_k-S_k(\sigma)}.
\]
Hence, by \Cref{lemma:mixture_martingale_comparison_with_pointwise} we get 
\begin{align}
\Pbb_{n,\sigma}(\widetilde S_k\ge 0)
&=
\mathbb E_{\mathbb Q_n}\biggl[
\left.\frac{\dd\Pbb_{n,\sigma}}{\dd\mathbb Q_n}\right|_{\mathcal F_k}
\mathbf 1_{\{\widetilde S_k\ge 0\}}
\biggr] \nonumber\\
&= \mathbb E_{\mathbb Q_n}\left[ e^{\widetilde S_k - S_k(\sigma)} \mathbf 1_{\{\widetilde S_k\ge 0\}} \right] \nonumber\\
% &\leq e^{C_\Pi}k^{d^2-1}\mathbb E_{\mathbb Q_n}\left[\mathbf 1_{\{\widetilde S_k\ge 0\}} \right] \nonumber\\
&\leq e^{\widetilde C}k^{d^2-1}\mathbb Q_n(\widetilde S_k\ge 0)\nonumber\\
&\le e^{\widetilde C}k^{d^2-1}c^k. \nonumber
\end{align}
Thus
\begin{align}
\mathbb E_{n,\sigma}\left[\sum_{j=1}^{T_n}\mathbf 1_{\{\widetilde S_{j-1}\ge 0\}}\right]
&\le
\sum_{j=1}^{\infty}\Pbb_{n,\sigma}(\widetilde S_{j-1}\ge 0)\nonumber\\
&\le
1+e^{\widetilde C}\sum_{j=2}^{\infty}(j-1)^{d^2-1}c^{j-1} \nonumber 
\\&=:C_+<\infty. \nonumber
\end{align}
where we used that $\sum_{j=1}^{T_n} \mathbf 1_{\{\widetilde S_{j-1}\ge 0\}} \le \sum_{j=1}^{\infty} \mathbf 1_{\{\widetilde S_{j-1}\ge 0\}}$, then applied Tonelli's theorem to exchange expectation and summation. Lastly, the term $j=1$ is bounded by $1$, while for $j\ge2$ we apply \Cref{lemma:alternative_lemma} \ref{lemma:alternative_lemma;part2}. The
remaining polynomially weighted geometric series converges since
$c\in(0,1)$, which concludes the proof.
\end{document}